\documentclass[journal, final, onecolumn, 10pt]{IEEEtran}

\usepackage[dvips]{graphicx}
\usepackage[cmex10]{amsmath}
\usepackage{amssymb}
\usepackage{array}
\usepackage{psfrag}
\usepackage[scriptsize]{subfigure}

\newcommand{\bbmat}{\begin{bmatrix} }
\newcommand{\ebmat}{\end{bmatrix} }
\newcommand{\Ber}{{\rm Ber}}

\newcommand{\eqdef}{\stackrel{\triangle}{=}}
\newcommand{\vari}[1]{{\mathbf{#1}}}

\newtheorem{definition}{Definition}
\newtheorem{theorem}{Theorem}
\newtheorem{lemma}{Lemma}
\newtheorem{corollary}{Corollary}
\newtheorem{remark}{Remark}

\graphicspath{{images/}}

\begin{document}

\title{Wyner-Ziv Coding over Broadcast Channels:\\ Digital Schemes}

\author{
Jayanth Nayak, Ertem Tuncel, Deniz G\"und\"uz\thanks{J.~Nayak was with the University of California, Riverside, CA. He is now with Mayachitra, Inc., Santa Barbara, CA. E-mail: nayak@mayachitra.com.}
\thanks{E.~Tuncel is with the University of California, Riverside, CA. E-mail: ertem@ee.ucr.edu.}
\thanks{D.~G\"und\"uz was with Princeton University and the Stanford University. He is now with the Centre Tecnol\'{o}gic de  Telecomunicacions de Catalunya (CTTC), 08860, Castelldefels, Barcelona, Spain. E-mail: deniz.gunduz@cttc.es}
\thanks{This work was presented in part at the Information Theory Workshop (ITW) 2008, Porto, Portugal.}}

\maketitle

\begin{abstract}
This paper addresses lossy transmission of a common source over a broadcast channel when there is correlated side information at the receivers, with emphasis on the quadratic Gaussian and binary Hamming cases. 
A digital scheme that combines ideas from the lossless version of the problem, i.e., Slepian-Wolf coding over broadcast channels, and dirty paper coding, is presented and analyzed. 
This scheme uses layered coding where the common layer information is intended for both receivers and the refinement information is destined only for one receiver.
For the quadratic Gaussian case, a quantity characterizing the overall quality of each receiver is identified in terms of channel and side information parameters. 
It is shown that it is more advantageous to send the refinement information to the receiver with ``better'' overall quality.
In the case where all receivers have the same overall quality, the presented scheme becomes optimal.
Unlike its lossless counterpart, however, the problem eludes a complete characterization.

%
%Among the many possible schemes, we shall pay particular attention, several digital schemes are proposed and their single-letter distortion tradeoffs are characterized. These schemes use layered coding where the common layer information is intended for both receivers and the refinement information is destined only for the receiver that is chosen using an appropriately defined ``combined'' channel/side information quality measure. When this quality is constant at each receiver, all the new schemes converge and become optimal. When the source and the channel bandwidths are equal, it is shown that one of the proposed schemes outperforms all the others as well as separate coding. For the quadratic Gaussian problem, it is also shown that if the combined quality criterion chooses the worse channel as the refinement receiver, then the same scheme also outperforms uncoded transmission. Unlike its lossless counterpart, however, the problem eludes a complete characterization.
\end{abstract}

\baselineskip .75cm

%%%%%%%%%%%%%%%%%%%%%%
%%%%%%%%%%%%%%%%%%%%%%
\section{Introduction}
\label{sec:Introduction}
%%%%%%%%%%%%%%%%%%%%%%
%%%%%%%%%%%%%%%%%%%%%%

Consider a sensor network of $K+1$ nodes taking periodic measurements of a common phenomenon. We study the communication scenario in which one of the sensors is required to transmit its measurements to the other $K$ nodes over a broadcast channel.
The receiver nodes are themselves equipped with side information unavailable to the sender, e.g., measurements correlated with the sender's data. This scenario, which is depicted in Figure~\ref{figr:Main}, can be of interest either by itself or as part of a larger scheme where all nodes are required to broadcast their measurements to all the other nodes.
Finding the capacity of a broadcast channel is a longstanding open problem, and thus, limitations of using separate source and channel codes in this scenario may never be fully understood. In contrast, a very simple joint source-channel coding strategy is optimal for the special case of {\em lossless} coding~\cite{Ertem}.
More specifically, it was shown in~\cite{Ertem} that in Slepian-Wolf coding over broadcast channels (SWBC), as the lossless case was referred to, for a given source $X$, side information $Y_1,\ldots,Y_K$, and a broadcast channel $p_{V_1\ldots V_K|U}$, lossless transmission (in the Shannon sense) is possible with $\kappa$ channel uses per source symbol if and only if there exists a channel input distribution $U$ such that
\begin{equation}
\label{eqtn:SWBC}
H(X|Y_k) \leq \kappa I(U;V_k)
\end{equation}
for $k=1,\ldots,K$. In the optimal coding strategy, every typical source word $X^n(i)$ is randomly mapped to a channel codeword $U^m(i)$, where $n$ and $m$ are so that $\kappa=\frac{m}{n}$. If~\eqref{eqtn:SWBC} is satisfied, there exists a channel codebook such that with high probability, there is a unique index $i$ for which $X^n(i)$ is jointly typical with the side information $Y^n_k$ and $U^m(i)$ is jointly typical with the channel output $V^m_k$ simultaneously, at any receiver $k$. This result exhibits some striking features which are worth repeating here.
\begin{enumerate}
\item[(i)] The optimal coding scheme is not separable in the classical sense, but consists of {\em separate components} that perform source and channel coding in a broader sense. This results in the separation of source and channel variables as in (\ref{eqtn:SWBC}).
\item[(ii)] If the broadcast channel is such that the same input distribution achieves capacity for all individual channels, then (\ref{eqtn:SWBC}) implies that one can utilize all channels at {\em full capacity}. Binary symmetric channels and Gaussian channels are the widely known examples of this phenomenon.
\item[(iii)] The optimal coding scheme does not explicitly involve {\em binning}, which is commonly used in network information theory. Instead, with the simple coding strategy of~\cite{Ertem}, each channel can be thought of as performing its own binning.
More specifically, the channel output $V_k^m$ at each receiver can be viewed as corresponding to a {\em virtual} bin\footnote{The bins can also be viewed as exponentially sized lists and a similar strategy that interprets the decoding as the intersection of exponentially sized lists was derived independently in~\cite{Laneman} and~\cite{Ertem}. Another alternative binning-based coding scheme that achieves the same performance using block Markov encoding and backward decoding can be found in~\cite{DenizITW2007}.} containing all source words $X^n(i)$ that map to channel codewords $U^m(i)$ jointly typical with $V_k^m$.
In general, the virtual bins can overlap and correct decoding is guaranteed by the size of the bins, which is about $2^{n[I(X;Y_k)-\epsilon]}$.
\end{enumerate}

\begin{figure}
\centering
\includegraphics[width=6in]{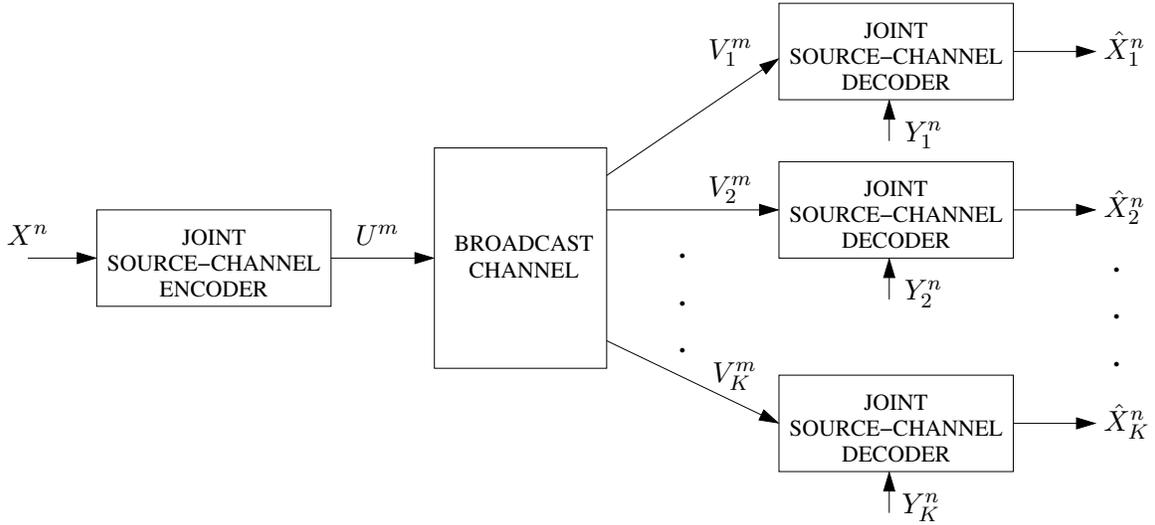}
\caption{Block diagram for Wyner-Ziv coding over broadcast channels.}\label{figr:Main}
\end{figure}

In this paper, we consider the general lossy coding problem in which the reconstruction of the source at the receivers need not be perfect. We shall refer to this problem setup as Wyner-Ziv coding over broadcast channels (WZBC). We present a coding scheme for this scenario and analyze its performance in the quadratic Gaussian and binary Hamming cases. This scheme uses ideas from SWBC~\cite{Ertem} and dirty paper coding (DPC)~\cite{Costa,Gelfand} as a starting point. The SWBC scheme is modified a) to allow quantization of the source, and b) to handle channel state information (CSI) at the encoder by using DPC. The modification with DPC is then employed in a layered transmission scheme with $K=2$ receivers, where there is common layer (CL) information destined for both receivers and refinement layer (RL) information meant for only one of the receivers. The channel codewords corresponding to the two layers are superposed and the resultant interference is mitigated using DPC. We shall briefly discuss other possible layered schemes obtained by varying the encoding and the decoding orders of the two layers and using successive coding or DPC to counteract the interference, although for the bandwidth matched Gaussian and binary Hamming cases, we observe that these variants perform worse.

DPC is used in this work in a manner quite different from the way it was used in~\cite{Caire}, which concentrated on sending private information to each receiver in a broadcast channel setting, where the information that forms the CSI and the information that is dirty paper coded are meant for different receivers. Therefore, although the DPC auxiliary codewords are decoded at one of the receivers, unlike in our scheme, this is of no use to that receiver. For our problem, this difference leads to an additional interplay in the choice of channel random variables. The DPC techniques in this work are most similar to those in~\cite{SutivongCover, WilsonSub}, where, as in our scheme, the CSI carries information about the source and hence decoding the DPC auxiliary codeword helps improve the performance. However, our results indicate a unique feature of DPC in the framework of WZBC. In particular, in our layered scheme, the optimal Costa parameter for the quadratic Gaussian problem turns out to be either 0 or 1. When it is 0, there is effectively no DPC, and when it is 1, the auxiliary codeword is identical to the channel input corrupted by the CSI. To the best of our knowledge, although the latter choice is optimal for binary symmetric channels, it has never been shown to be optimal for a Gaussian channel in a scenario considered before.

When an appropriately defined ``combined'' channel and side information quality is constant at each receiver, the new scheme is shown to be optimal in the quadratic Gaussian case. We also derive conditions for the same phenomenon to occur in the binary Hamming case, although the expressions are not as elegant as in the quadratic Gaussian problem. Unlike in~\cite{Ertem}, however, the scheme that we derive is not always optimal. A simple alternative approach is to separate the source and channel coding. Both Gaussian and binary symmetric broadcast channels are degraded. Hence their capacity regions are known~\cite{CoverThomas} and further, there is no loss of optimality in confining ourselves to two layer source coding schemes. The corresponding source and side information pairs are also degraded. Although a full characterization of the rate-distortion performance is available for the quadratic Gaussian case~\cite{TianDiggavi}, only a partial characterization is available for the binary Hamming problem~\cite{SteinbergMerhav, TianDiggavi}. In any case, we obtain an achievable distortion tradeoff of separate source and channel coding by combining the known rate-distortion results with the capacity results. For the quadratic Gaussian problem, we show that our scheme always performs at least as well as separate coding. The same phenomenon is numerically observed for the binary Hamming case.

For the two examples we consider, a second alternative is uncoded transmission if there is no bandwidth expansion or compression. This scheme is optimal in the absence of side information at the receivers in both the quadratic Gaussian and binary Hamming cases. However, in the presence of side information, the optimality may break down. We show that, depending on the quality of the side information, our scheme can indeed outperform uncoded transmission as well. In particular, if the combined quality criterion chooses the worse channel as the refinement receiver (because it has much better side information), then our layered scheme outperforms uncoded transmission for the quadratic Gaussian problem.

The paper is organized as follows. In Section~\ref{sect:Prelims}, we formally define the problem and present relevant past work. Our main results are presented in Section~\ref{sect:Basic} and Section~\ref{sect:Layered}, namely the extensions of the scheme in~\cite{Ertem} that we develop for the lossy scenario.
We then analyze a layered scheme in particular for the quadratic Gaussian and binary Hamming cases in Sections~\ref{sctn:QuadraticGaussian} and~\ref{sctn:BinaryHamming}, respectively. For these cases, we compare the derived schemes with separate source and channel coding, and with uncoded transmission. Section~\ref{sctn:Conclusions} concludes the paper by summarizing the results and pointing to future work.

%%%%%%%%%%%%%%%%%%%%%%%%%%%%%%%%%
%%%%%%%%%%%%%%%%%%%%%%%%%%%%%%%%%
\section{Background and Notation}
\label{sect:Prelims}
%%%%%%%%%%%%%%%%%%%%%%%%%%%%%%%%%
%%%%%%%%%%%%%%%%%%%%%%%%%%%%%%%%%

Let $(X, Y_1, \ldots, Y_K) \in \mathcal{X}\times\mathcal{Y}_1\times\cdots\times\mathcal{Y}_K$ be random variables denoting a source with independent and identically distributed (i.i.d.) realizations. Source $X$ is to be transmitted over a memoryless broadcast channel defined by $p_{V_1\cdots V_K|U}(v_1,\ldots,v_K|u),$\ $\; u\in\mathcal{U}, v_k\in\mathcal{V}_k, k=1,\ldots K$. Decoder $k$ has access to side information $Y_k$ in addition to the channel output $V_k$.
Let single-letter distortion measures $d_k:\mathcal{X}\times \hat{\mathcal{X}}_k\to [0, \infty)$ be defined at each receiver, i.e.,
\[
d_k(x^n,\hat{x}_k^n) = \frac{1}{n} \sum_{j=1}^n d_k(x_j, \hat{x}_{kj})
\]
for $k=1,\ldots,K$.

\begin{definition}
An $(m,n,\phi, \psi_1,\ldots, \psi_K)$ code consists of an encoder
\[
\phi:\mathcal{X}^n\to \mathcal{U}^m
\]
and decoders at each receiver
\[
\psi_k:\mathcal{V}_k^m \times \mathcal{Y}_k^n \to \hat{\mathcal{X}}_k^n \; .
\]
The rate of the code is $\kappa=\frac{m}{n}$ channel uses per source symbol.
\end{definition}

\begin{definition}
A distortion tuple $(D_1, \ldots,D_K)$ is said to be achievable at a rational rate $\kappa$  if for every $\epsilon>0$, there exists $n_0$ such that for all integers $m>0,n>n_0$ with $\frac{m}{n} = \kappa$, there exists an $(m, n, \phi, \psi_1, \ldots,\psi_K)$ code satisfying
\[
\frac{1}{n} \mbox{E}\left[d_k(X^n,\hat{X}_k^n)\right] \leq D_k + \epsilon
\]
where $\hat{X}_{k}^n = \psi_k( V_k^m, Y_k^n)$ and $V_k^m$ denotes the channel output corresponding to $\phi(X^n)$.
\end{definition}

In this paper, we present some general WZBC techniques and derive the corresponding achievable distortion regions. We study the performance of these techniques for the following cases.
\begin{itemize}
\item {\em Quadratic Gaussian}:
All source and channel variables are real-valued, and we use the notation $\vari{A}$ to denote the variance of any Gaussian random variable $A$.
The source and side information are jointly Gaussian and the channels are additive white Gaussian, i.e., $V_k = U + W_k$ where $W_k$ is Gaussian and $W_k$ is independent of $U$.
There is an input power constraint on the channel:
\[
\frac{1}{m} \sum_{j=1}^m \mbox{E}[ (U_j)^2 ] \leq P
\]
where $U^m = \phi(X^n)$.  Without loss of generality, we assume that $\vari{X}=\vari{Y}_1=\cdots=\vari{Y}_K=1$ and $Y_k=\rho_k X + N_k$ with $N_k \perp X$ and $\rho_k>0$.
Thus, $\vari{N}_k = 1-\rho_k^2$, denotes the mean squared-error in estimating $X$ from $Y_k$, or equivalently, $Y_k$ from $X$ since $\vari{X}=\vari{Y}_k$. Reconstruction quality is measured by squared-error distance: $d_k(x,\hat{x}_k)= (x-\hat{x}_k)^2$.
\item {\em Binary Hamming}: All source and channel alphabets are binary. The source is $\Ber(\tfrac{1}{2})$, where $\Ber(\epsilon)$ denotes the Bernoulli distribution with $P[1]=\epsilon$. The channels are binary symmetric with transition probabilities $p_k$, i.e., $V_k = U_k \oplus W_k$ where $W_k \sim \Ber(p_k)$ and $W_k$ and $U_k$ are independent with $\oplus$ denoting modulo 2 addition (or the XOR operation). The side information sequences at the receivers are also noisy versions of the source corrupted by passage through virtual binary symmetric channels; that is, $Y_k = X_k \oplus N_k$ with $N_k \sim \Ber(\beta_k)$ and $N_k$ and $X_k$ are independent. Reconstruction quality is measured by Hamming distance: $d_k(x,\hat{x}_k) = x\oplus\hat{x}_k$.

\end{itemize}
The problems considered in~\cite{ Kramer, Reznic, Ertem} can all be seen as special cases of the WZBC problem. However, the quadratic Gaussian and the binary Hamming cases with non-trivial side information have never, to our knowledge, been analyzed before. Nevertheless, separate source and channel coding and uncoded transmission are obvious strategies. We shall evaluate the performance of these alternative strategies and present numerical comparisons with our proposed scheme.

%%%%%%%%%%%%%%%%%%%%%%%%%%%%%%%%%%%%%%%%%%%%%%%%%%%%%%%%%%
\subsection{Wyner-Ziv Coding over Point-to-Point Channels}
\label{subs:P2P}
%%%%%%%%%%%%%%%%%%%%%%%%%%%%%%%%%%%%%%%%%%%%%%%%%%%%%%%%%%

Before analyzing the WZBC problem in depth, we shall briefly discuss known results for Wyner-Ziv coding over a point-to-point channel, i.e., the case $K=1$. Since $K=1$, we shall drop the subscripts that relate to the receiver.
The Wyner-Ziv rate-distortion performance is characterized in~\cite{WZiv} as
\begin{equation}
\label{eqtn:WZDistortionRate}
D^{WZ}(R) \eqdef \min_{ \begin{array}{c} Z,g: Y-X-Z \\ I(X;Z|Y)\leq R\end{array}} \mbox{E} \left[d(X,g(Z,Y))\right] \; .
\end{equation}
where $Z\in{\cal Z}$ is an auxiliary random variable, and the capacity of the channel $p_{V|U}$ is well-known (cf.~\cite{CoverThomas}) to be
\[
C = \max_U I(U;V) \; .
\]
It is then straightforward to conclude that combining separate source and channel codes yields the distortion
\begin{equation}
\label{eqtn:WZDistortion}
D = D^{WZ}(\kappa C).
\end{equation}
On the other hand, a converse result in~\cite{ShamaiVerdu} shows that even by using joint source-channel codes, one cannot improve the distortion performance further than (\ref{eqtn:WZDistortion}).

We are further interested in the evaluation of $D^{WZ}(R)$, as well as in the test channels achieving it, for the quadratic Gaussian and binary Hamming cases. We will use similar test channels in our WZBC schemes.

%%%%%%%%%%%%%%%%%%%%%%%%%%%%%%%%%%%
\subsubsection{Quadratic Gaussian}
It was shown in~\cite{Wyner} that the optimal backward test channel is given by
\[
X = Z + S
\]
where $Z$ and $S$ are independent Gaussians.
For the rate we have~\footnote{All logarithms are base 2.}
\begin{equation}
\label{eqtn:WZR_Gaussian}
R \geq I(X;Z|Y) = \frac{1}{2}\log \left( 1-\vari{N} + \frac{\vari{N}}{\vari{S}}\right) \; .
\end{equation}
The optimal reconstruction is a linear estimate $g(Z,Y) = \frac{\vari{Z}\vari{N}}{\vari{Z}(1-\rho^2\vari{Z})} Z + \frac{\rho\vari{Z}(1-\vari{Z})}{\vari{Z}(1-\rho^2\vari{Z})} Y$, which yields the distortion
\begin{equation}
\label{eqtn:WZD_Gaussian}
\mbox{E}[d(X,g(Z,Y))] = \frac{\vari{N}}{1-\vari{N} + \frac{\vari{N}}{\vari{S}}}
\end{equation}
and therefore,
\begin{equation}
\label{eqtn:WZDR_Gaussian}
D^{WZ}(R) = \vari{N} 2^{-2R} \; .
\end{equation}

%%%%%%%%%%%%%%%%%%%%%%%%%%%%%%
\subsubsection{Binary Hamming}
It was implicitly shown in~\cite{WZiv} that the optimal auxiliary random variable $Z\in{\cal Z}=\{0,1,\lambda\}$ is given by
\[
Z = E \circ (X\oplus S)
\]
where $X, E, S$ are all independent, $E$ and $S$ are Ber($q$) and Ber($\alpha$) with $0\leq q\leq 1$ and $0\leq\alpha\leq \frac{1}{2}$, respectively, and $\circ$ is an erasure operator, i.e.,
\[
a \circ b = \begin{cases}
\lambda & a=0 \\
b & a=1
\end{cases} \; .
\]
This choice results in
\begin{equation}
\label{eqtn:WZR_Binary}
I(X;Z|Y) = q r(\alpha,\beta)
\end{equation}
where
\[
r(\alpha,\beta) = H_2(\alpha \star \beta) - H_2(\alpha)
\]
with $\star$ denoting the binary convolution, i.e.,  $a\star b = (1-a)b+a(1-b)$, and $H_2$ denoting the binary entropy function, i.e.,
\[
H_2(p) = -p\log p -(1-p)\log(1-p).
\]
It is easy to show that when $0\leq\alpha,\beta\leq \frac{1}{2}$, $r(\alpha,\beta)$ is increasing in $\beta$ and decreasing in $\alpha$.

Since $\mbox{E}[d(X,g(Z,Y))]=\Pr[X\neq g(Z,Y))]$ and $X\sim$Ber($\frac{1}{2}$), the corresponding optimal reconstruction function $g$ boils down to a maximum likelihood estimator given by
\begin{eqnarray*}
g(z,y) & = & \arg\max_x \, p_{Y Z|X}(y,z|x) \\
& = & \arg\max_x \, p_{Z|X}(z|x)p_{Y|X}(y|x) \\
& = & \begin{cases}
y & z = \lambda \mbox{ or } z = y\\
z & z \neq \lambda, z \neq y \mbox { and } \beta>\alpha \\
y & z \neq \lambda, z \neq y \mbox { and } \beta\leq\alpha
\end{cases} \; .
\end{eqnarray*}
The resultant distortion is given by
\begin{equation}
\label{eqtn:WZD_Binary}
\mbox{E}[d(X,g(Z,Y))] = q \min\{\alpha,\beta\} + (1-q)\beta
\end{equation}
implying together with (\ref{eqtn:WZR_Binary}) that
\begin{equation}
\label{eqtn:WZDR_Binary}
D^{WZ}(R) = \min_{ \begin{array}{c} 0\leq q \leq 1 ,0\leq \alpha\leq \beta:\\
q \, r(\alpha,\beta) \leq R \end{array}} \bigg[ q \alpha + (1-q)\beta \bigg]
\end{equation}
where the extra constraint $\alpha\leq \beta$ is imposed because $\alpha>\beta$ is a provably suboptimal choice.
It also follows from the discussion in~\cite{WZiv} that there exists a critical rate $R_0(\beta)$ above which the optimal test channel assumes $q=1$ and $0\leq\alpha\leq\alpha_0(\beta)\leq \beta$, and below which it assumes $\alpha=\alpha_0(\beta)$ and $0\leq q< 1$. The reason why we discussed other values of $(q,\alpha)$ above is because we will use the test channel in its most general form in all WZBC schemes.

%%%%%%%%%%%%%%%%%%%%%%%%%%%%%%%%%%%%%%%%%%%%%%%%%%%%
\subsection{A Trivial Converse for the WZBC Problem}
%%%%%%%%%%%%%%%%%%%%%%%%%%%%%%%%%%%%%%%%%%%%%%%%%%%%

At each terminal, no WZBC scheme can achieve a distortion less than the minimum distortion achievable by ignoring the other terminals.
Thus,
\begin{equation}
\label{eqtn:DWZ}
D_k \geq D_k^{WZ}(\kappa C_k)
\end{equation}
where $C_k$ is the capacity of channel $k$.
For the source-channel pairs we consider, (\ref{eqtn:DWZ}) can be further specialized.
For the quadratic Gaussian case, we obtain using (\ref{eqtn:WZDR_Gaussian}) and
\[
C_k = \frac{1}{2}\log \left(1+ \frac{P}{\vari{W}_k}\right)
\]
that
\begin{equation}
\label{eqtn:TrivialGaussianD}
D_k \geq \frac{\vari{N}_k}{(1+\frac{P}{\vari{W}_k})^{\kappa}} \;.
\end{equation}
For the binary Hamming case, using (\ref{eqtn:WZDR_Binary}) and $C_k=1-H_2(p_k)$, the converse becomes
\[
D_k \geq \min_{ \begin{array}{c} 0\leq q \leq 1 ,0\leq \alpha\leq \beta_k:\\
q\, r(\alpha,\beta) \leq \kappa[1 - H_2(p_k)] \end{array}} q \alpha + (1-q)\beta_k.
\]

%%%%%%%%%%%%%%%%%%%%%%%%%%%%%%%%%%%%%%%%%%%%%%%
\subsection{Separate Source and Channel Coding}
%%%%%%%%%%%%%%%%%%%%%%%%%%%%%%%%%%%%%%%%%%%%%%%

For a general source and channel pair, the source and channel coding problems are extremely challenging. The set of all achievable rate triples (common and two private rates) for general broadcast channels are not known. The corresponding source coding problem has not been explicitly considered in previous work either. But there is considerable simplification in the quadratic Gaussian and binary Hamming cases since the channel and the side information are degraded in both cases: we can assume that one of the two Markov chains, $U-V_1-V_2$ or $U-V_2-V_1$, holds (for arbitrary channel input $U$) for the channel, and similarly either $X-Y_1-Y_2$ or $X-Y_2-Y_1$ holds for the source. The capacity region for degraded broadcast channels is fully known. In fact, since any information sent to the weaker channel can be decoded by the stronger channel, we can assume that no private information is sent to the weaker channel. As a result, two layer source coding, which has been considered in~\cite{SteinbergMerhav, TianDiggavi, TianDiggavi2}, is sufficiently general.

To be able to analyze $U-V_1-V_2$ and $U-V_2-V_1$ simultaneously, we denote the random variables, rates, and distortion levels associated with the $g$ood channel by the subscript $g$ and those associated with the $b$ad one by $b$, i.e., the channel variables always satisfy $U-V_g-V_b$ where $g$ is either 1 or 2 and $b$ takes the other value. Let $\mathcal{C}(\kappa)$ denote the capacity region for $\kappa$ channel uses, i.e., the region of all pairs of total rates that can be simultaneously decoded by each receiver. As shown in~\cite{Bergmans, Gallager}, ${\cal C}(\kappa)$ is the convex closure of all $(R_b,R_g)$ such that there exist a channel input $U\in \mathcal{U}$ and an auxiliary random variable $U_b\in \mathcal{U}_b$ satisfying $U_b-U-V_g-V_b$, the power constraint (if any) $\mbox{E}[U^2]\leq P$, and
\begin{eqnarray}
\label{eqtn:SeparateCapacity1}
R_b & \leq & \kappa I(U_b;V_b) \\
\label{eqtn:SeparateCapacity2}
R_g & \leq & \kappa [I(U_b;V_b) + I(U;V_g|U_b)] \; .
\end{eqnarray}

Let ${\cal R}(D_b, D_g)$ be the set of total rates that must be sent to each source decoder to enable the receivers to reconstruct the source within the respective distortions $D_b$ and $D_g$. A distortion pair $(D_b, D_g)$ is achievable by separate source and channel coding with $\kappa$ channel uses per source symbol if and only if
\[
{\cal R}(D_b,D_g) \cap {\cal C}(\kappa) \neq \emptyset \; .
\]
Note that we use cumulative rates at the good receiver.

Despite the simplification brought by degraded side information, there is no known complete single-letter characterization of ${\cal R}(D_b,D_g)$ for all sources and distortion measures when $X-Y_b-Y_g$.
Let ${\cal R}^*(D_b,D_g)$ be defined as the convex closure of all $(R_b,R_g)$ such that there exist source auxiliary random variables $(Z_b, Z_g)\in \mathcal{Z}_b\times\mathcal{Z}_g$ with either $(Y_b,Y_g)-X-Z_b-Z_g$ or $(Y_b,Y_g)-X-Z_g-Z_b$, and reconstruction functions $g_k:\mathcal{Z}_k\times\mathcal{Y}_k\to\hat{\mathcal{X}}$ satisfying
\begin{equation}
\label{eqtn:SeparateDistortion}
\mbox{E}[d_k(X, g_k(Z_k,Y_k))] \leq D_k
\end{equation}
for $k=b,g$, and
\begin{eqnarray}
\label{eqtn:SeparateRate1}
R_b & \geq & I(X;Z_b|Y_b) \\
\label{eqtn:SeparateRate2}
R_g & \geq &  \begin{cases}
I(X;Z_b|Y_b) +[I(X;Z_g|Y_g) - I(X;Z_b|Y_g)]^+ & \mbox{ if } X-Y_g-Y_b \\
I(X;Z_g|Y_g) + [I(X;Z_b|Y_b) - I(X;Z_g|Y_b)]^+ & \mbox{ if } X-Y_b-Y_g
\end{cases} \; .
\end{eqnarray}
It was shown in~\cite{SteinbergMerhav} that ${\cal R}(D_b,D_g)={\cal R}^*(D_b,D_g)$ when $X-Y_g-Y_b$.
On the other hand, \cite{TianDiggavi} showed that even when $X-Y_b-Y_g$, ${\cal R}(D_b,D_g)={\cal R}^*(D_b,D_g)$ for the quadratic Gaussian problem.
For all other sources and distortion measures, we only know ${\cal R}(D_b,D_g) \supset {\cal R}^*(D_b,D_g)$ in general when $X-Y_b-Y_g$. We shall present explicit expressions for the complete tradeoff in the quadratic Gaussian case in Section~\ref{sctn:QuadraticGaussian} and an achievable tradeoff for the binary Hamming case in  Section~\ref{sctn:BinaryHamming}.

%%%%%%%%%%%%%%%%%%%%%%%%%%%%%%%%%
\subsection{Uncoded Transmission}
%%%%%%%%%%%%%%%%%%%%%%%%%%%%%%%%%

In the bandwidth-matched case, i.e., when $\kappa=1$, if the source and channel alphabets are compatible, uncoded transmission is a possible strategy. For the quadratic Gaussian case, the distortion achieved by uncoded transmission is given by
\begin{equation}
\label{eqtn:WZ_UncodedGaussian}
D_k = \frac{\vari{N}_k\vari{W}_k}{\vari{W}_k + \vari{N}_k P}
\end{equation}
for $k=1,2$. This, in turn, is also because the channel is the same as the test channel up to a scaling factor. More specifically, when $\sqrt{P}X$ is transmitted and corrupted by noise $W_k$, one can write $X=Z_k+S_k$ with $S_k\perp Z_k$, where $Z_k$ is an appropriately scaled version of the received signal $\sqrt{P}X+W_k$ and
\[
\vari{S}_k = \frac{\vari{W}_k}{\vari{W}_k+P} \; .
\]
Substituting this into (\ref{eqtn:WZD_Gaussian}) then yields (\ref{eqtn:WZ_UncodedGaussian}).
Comparing with \eqref{eqtn:TrivialGaussianD}, we note that \eqref{eqtn:WZ_UncodedGaussian} achieves $D^{WZ}_k(C_k)$ only when $\vari{N}_k=1$ or when $\vari{W}_k\to\infty$, which, in turn, translate to trivial $Y_k$ or zero $C_k$, respectively.

For the binary Hamming case, this strategy achieves the distortion pair
\begin{equation}
\label{eqtn:WZ_UncodedBinary}
D_k = \min\{ p_k, \beta_k\}
\end{equation}
for $k=1,2$.
That is because the channel is the same as the test channel that achieves $D^{WZ}(R)$ with $q=1$. The distortion expression in (\ref{eqtn:WZ_UncodedBinary}) then follows using (\ref{eqtn:WZD_Binary}).
One can also show that \eqref{eqtn:WZ_UncodedBinary} coincides with $D^{WZ}_k(C_k)$ only when $\beta_k=\frac{1}{2}$ or $p_k=\frac{1}{2}$.
Once again, these respectively correspond to trivial $Y_k$ and zero $C_k$.

%%%%%%%%%%%%%%%%%%%%%%%%%%%%
%%%%%%%%%%%%%%%%%%%%%%%%%%%%
\section{Basic WZBC Schemes}
\label{sect:Basic}
%%%%%%%%%%%%%%%%%%%%%%%%%%%%
%%%%%%%%%%%%%%%%%%%%%%%%%%%%

In this section, we present the basic coding schemes that we shall then develop into the schemes that form the main contributions of this paper.
In what follows, we only present code constructions for discrete sources and channels. The constructions can be extended to the continuous case in the usual manner. Our coding arguments rely heavily on the notion of typicality. Given a random variable $X\sim P_X(x),$ defined over a discrete alphabet $\mathcal{X}$ the typical set at block length $n$ is defined as~\cite{AlonRoche}
\[
{\cal T}_{\delta}^n(X) \triangleq \left\{ x^n\in\mathcal{X}^n: \left| \frac{N(a|x^n)}{n} - P_X(a)\right| \leq \delta P_X(a), \forall a\in\mathcal{X}\right\}
\]
where $N(a|x^n)$ denotes the number of times $a$ appears in $x^n$.

The first scheme, termed Common Description Scheme (CDS), is a basic extension of the scheme in~\cite{Ertem} where the source is first quantized before transmission over the channel. Even though our layered schemes are constructed for the case of $K=2$ receivers, CDS can be utilized for any $K\geq 2$. Unlike in~\cite{Ertem}, where typical source words are placed in one-to-one correspondence with a channel codebook, the source words are first mapped to quantized versions and it is these quantized versions that are mapped to the channel codebook. Like~\cite{Ertem}, there is no explicit binning, but the channel performs \emph{virtual} binning. Before discussing the performance of the CDS, we shall present an extension of the CDS for a more general coding problem.

Suppose that there is CSI available solely at the encoder, i.e., the broadcast channel is defined by the transition probability $p_{V_1V_2|US}(v_1,v_2|u,s)$ and the CSI $S^m\in {\cal T}_{\eta}^m(S)$ with some $\eta>0$, where $S$ is some fixed distribution defined on the CSI alphabet $\mathcal{S}$, is available non-causally at the encoder. Given a source and side information at the decoders $(X, Y_1, Y_2)$, codes $(m,n,\phi,\psi_1,\psi_2)$ and achievability of distortion pairs is defined as in the WZBC scenario except that the encoder now takes the form $\phi:\mathcal{X}^n\times\mathcal{S}^m\to\mathcal{U}^m$. The following theorem characterizes the performance of an extension of the CDS, which we term CDS with DPC.

\begin{theorem}\label{thrm:S0_DPC}
A distortion pair $(D_1,\ldots,D_K)$ is achievable at rate $\kappa$ if there exist random variables $Z\in\mathcal{Z}$, $T\in\mathcal{T}, U\in \mathcal{U}$ and functions $g_k:\mathcal{Z}\times\mathcal{Y}_k\to \hat{\mathcal{X}}$ with $(Y_1, \ldots, Y_K)-X-Z$ and $T-(U,S)-(V_1,\ldots,V_K)$ such that
\begin{align}
I(X;Z|Y_k) &< \kappa \big[I(T; V_k) - I(T;S)\big] \label{eqtn:Ineq_S0_DPC}\\
\mbox{E}[d_k(X,g_k(Z,Y_k))] &\leq D_k \label{eqtn:Ineq_S0_DPC_D}
\end{align}
for $k=1,\ldots,K$.
\end{theorem}
\begin{proof}

The code construction is as follows. For fixed $\delta,\delta',\delta''>0$, a source codebook $\mathcal{C}_Z\triangleq \{z^n(i), i=1,\dotsc,M\}$ is chosen from ${\cal T}_{\delta}^n(Z)$. A set of $M$ bins $\mathcal{C}_T(i)=\{ t^m(i,j), j=1,\dotsc,M'\}$, where each $t^m(i,j)$ is chosen randomly at uniform from ${\cal T}_{\delta}^m(T)$, is also constructed. Given a source word $X^n$ and CSI $S^m$, the encoder tries to find a pair $(i^*,j^*)$ such that $(X^n, z^n(i^*))\in {\cal T}_{\delta'}^n(X,Z)$ and $(S^m,t^m(i^*,j^*))\in {\cal T}_{\delta'}^m(S,T)$. If it is unsuccessful, it declares an error. If it is successful, the channel input is drawn from the distribution $\prod_{l=1}^m p_{U|TS}(u_l|t_l(i^*,j^*),S_l)$.
At terminal $k$, the decoder goes through all pairs $(i,j)\in\{1,\dotsc,M\}\times\{1,\dotsc,M'\}$ until it finds the first pair satisfying $(Y_k^n, z^n(i))\in {\cal T}_{\delta''}^n(Y_k,Z)$ and $(V_k^m,t^m(i,j))\in {\cal T}_{\delta''}^m(V_k,T)$ simultaneously. If there is no such pair, the decoder sets $i=1,j=1$. Once $(i,j)$ is decided, coordinate-wise reconstruction is performed using $g_k$ with $Y_k^n$ and $z^n(i)$.

We define the error events as
\begin{eqnarray*}
{\cal E}_1 & = & \bigg< \forall (i,j), \;\; \mbox{either } (X^n,z^n(i))\not\in {\cal T}_{\delta'}^n(X,Z) \; \mbox{ or } \; (S^m,t^m(i,j))\not\in {\cal T}_{\delta'}^m(S,T) \bigg> \\
{\cal E}_2(k) & = & \bigg< (Y_k^n,z^n(i^*))\not\in {\cal T}_{\delta''}^n(Y_k,Z) \bigg> \\
{\cal E}_3(k) & = & \bigg< (V^m_k,t^m(i^*,j^*))\not\in {\cal T}_{\delta''}^m(V_k,T) \bigg> \\
{\cal E}_4(k) & = & \bigg< \exists(i\neq i^*,j), \;\; (Y_k^n,z^n(i))\in {\cal T}_{\delta''}^n(Y_k,Z)\; \mbox{ and } \;(V^m_k,t^m(i,j))\in {\cal T}_{\delta''}^m(V_k,T) \bigg> \; .
\end{eqnarray*}
Using standard typicality arguments, it can be shown that for fixed $\delta,\delta',\delta''$, if
\[
M \geq 2^{n[I(X;Z)+\epsilon_1(\delta,\delta',\delta'')]}
\]
and
\[
M' \geq 2^{m[I(S;T)+\epsilon_1(\delta,\delta',\delta'')]}
\]
then $\Pr[{\cal E}_1]<\epsilon$, and that $\Pr[{\cal E}_2(k)]<\epsilon$ and $\Pr[{\cal E}_3(k)]<\epsilon$ for any $\epsilon>0$ and large enough $n$.
Similarly, it follows that if
\[
M \leq 2^{n[I(X;Z)+2\epsilon_1(\delta,\delta',\delta'')]}
\]
and
\[
M' \leq 2^{m[I(S;T)+2\epsilon_1(\delta,\delta',\delta'')]}
\]
then
\begin{eqnarray*}
\Pr[{\cal E}_4(k)] & \leq & M\cdot M'\cdot 2^{-n[I(Y_k;Z)-\epsilon_2(\delta,\delta',\delta'')]} 2^{-m[I(T;V_k)-\epsilon_2(\delta,\delta',\delta'')]} \\
& = & M\cdot M'\cdot 2^{-n[I(Y_k;Z)+\kappa I(T;V_k)-(\kappa+1)\epsilon_2(\delta,\delta',\delta'')]} \\
& \leq & 2^{n[I(X;Z|Y_k)-\kappa \{I(T;V_k)-I(S;T)\}+(\kappa+1)\epsilon_2(\delta,\delta',\delta'')+2(\kappa+1)\epsilon_1(\delta,\delta',\delta'')]} \; .
\end{eqnarray*}
This probability also vanishes if $\delta,\delta',\delta''\to\infty$ thanks to (\ref{eqtn:Ineq_S0_DPC}).
This completes the proof.
\end{proof}

Note that, if $S$ is a trivial random variable, independent of the channel, the scenario becomes the original WZBC setup and CDS with DPC becomes CDS. By equating $T$ and $U$, we obtain the following corollary that characterizes the performance of the CDS.

\begin{corollary}\label{corr:S0}
A distortion tuple $(D_1,\ldots, D_K)$ is achievable at rate $\kappa$ for the WZBC problem if there exist random variables $Z\in\mathcal{Z}$, $U\in \mathcal{U}$ and functions $g_k:\mathcal{Z}\times\mathcal{Y}_k\to \hat{\mathcal{X}_k}$ with $(Y_1, \ldots,Y_K)-X-Z$ such that
\begin{align}
I(X;Z|Y_k) &< \kappa I(U; V_k)  \label{eqtn:Ineq_S0}\\
\mbox{E}[d_k(X,g_k(Z,Y_k))] &\leq D_k \label{eqtn:Ineq_S0_D}
\end{align}
for $k=1,\ldots,K$.
\end{corollary}

\begin{corollary}
\label{corr:Scheme0_DPC}
The coding scheme in the proof of Theorem~\ref{thrm:S0_DPC} can also decode $t^m(i^*,j^*)$ successfully.
\end{corollary}
\begin{proof}
Define
\[
{\cal E}_5(k) = \bigg< \exists j\neq j^*, \;\; (V^m_k,t^m(i^*,j))\in {\cal T}_{\delta''}^m(V_k,T) \bigg> \; .
\]
It then suffices to show that $\Pr[{\cal E}_5(k)]<\epsilon$ for large enough $n$. Indeed, since $I(T;V_k)-I(S;T)>0$,
\begin{eqnarray*}
\Pr[{\cal E}_5(k)] & \leq &  M' 2^{-m[I(T;V_k)-\epsilon_2(\delta,\delta',\delta'')]}  \\
& \leq &  2^{-m[I(T;V_k)-I(S;T)-\epsilon_2(\delta,\delta',\delta'')-2\epsilon_1(\delta,\delta',\delta'')]}  \\
& \leq & \epsilon \; .
\end{eqnarray*}
The assumption $I(T;V_k)-I(S;T)>0$ is not restrictive at all, because otherwise no information can be delivered to terminal $k$ to begin with.
\end{proof}

The significance of Corollary~\ref{corr:Scheme0_DPC} is that decoding $t^m(i^*,j^*)$ provides information about the CSI $S^m$. This information, in turn, will be very useful in our layered WZBC schemes where the CSI is self-imposed and related to the source $X^n$ itself.

Examining the proof of Theorem \ref{thrm:S0_DPC}, we notice an apparent separation between source and channel coding in that the source and channel codebooks are independently chosen. Furthermore, successful transmission is possible as long as the source coding rate for each terminal is less than the corresponding channel coding rate for a common channel input.
However, the decoding must be jointly performed and neither scheme can be split into separate stand-alone source and channel codes. Nevertheless, due to the quasi-independence of the source and channel codebooks we shall refer to source codes and channel codes separately when we discuss layered WZBC schemes.
This quasi-separation was shown to be optimal for the SWBC problem and was termed {\em operational separation} in~\cite{Ertem}.

%%%%%%%%%%%%%%%%%%%%%%%%%%%%%%
%%%%%%%%%%%%%%%%%%%%%%%%%%%%%%
\section{A Layered WZBC Scheme}
\label{sect:Layered}
%%%%%%%%%%%%%%%%%%%%%%%%%%%%%%
%%%%%%%%%%%%%%%%%%%%%%%%%%%%%%

In this section, we focus on the case of $K=2$ receivers. In CDS, the same information is conveyed to both receivers. However, since the side information and channel characteristics at the two receiving terminals can be very different, we might be able to improve the performance by layered coding, i.e., by not only transmitting a common layer (CL) to both receivers but also additionally transmitting a refinement layer (RL) to one of the two receivers.
The resultant interference between the CL and RL can then be mitigated by successive decoding or by dirty paper encoding.
Since there are two receivers, we are focusing on coding with only two layers because intuitively, more layers targeted for the same receiver can only degrade the performance.

Unless the better channel also has access to better side information, it is not straightforward to decide which receiver should receive only the CL and which should additionally receive the RL. We shall therefore refer to the decoders as the CL decoder and the RL decoder (which necessarily also decodes the CL) instead of using the subscripts $1$ and $2$. For the quadratic Gaussian problem, we will later develop an analytical decision tool. For all other sources and channels, one can combine the distortion regions resulting from the two choices, namely, CL decoder $=1$ and RL decoder $=2$ and vice versa. For ease of exposition, for a given choice of CL and RL decoders, we also rename the source and channel random variables by replacing the subscripts 1 and 2 by $c$ (for random variables corresponding to the CL information or to the CL decoder) and $r$ (for random variables corresponding to the RL information or to the receiver that decodes \emph{both} CL and RL).

As mentioned earlier, the inclusion of an RL codeword changes the effective channel observed while decoding the CL. It is on this modified channel that we send the CL using CDS or CDS with DPC, and the respective channel rate expressions in \eqref{eqtn:Ineq_S0} and \eqref{eqtn:Ineq_S0_DPC} must be modified in a manner that we describe in the following subsections where we also present the capacity of the effective channel for transmitting the RL. Each possible order of channel encoding and decoding (at the RL decoder) leads to a different scheme. We shall concentrate on the scheme that has the best performance among the four in the Gaussian and binary Hamming cases, deferring a discussion of the other three to Appendix~\ref{subs:App:OtherSchemes}. In this scheme, illustrated in Figure~\ref{figr:RCCR}, the CL is coded using CDS with DPC with the RL codeword acting as CSI. We shall refer to this scheme as the Layered Description Scheme (LDS). We characterize the source and channel coding rates for LDS in the following. We will only sketch the proofs of the theorems, as they rely only on CDS with DPC, and other standard tools.

\begin{figure}
\centering
\includegraphics[width=15cm]{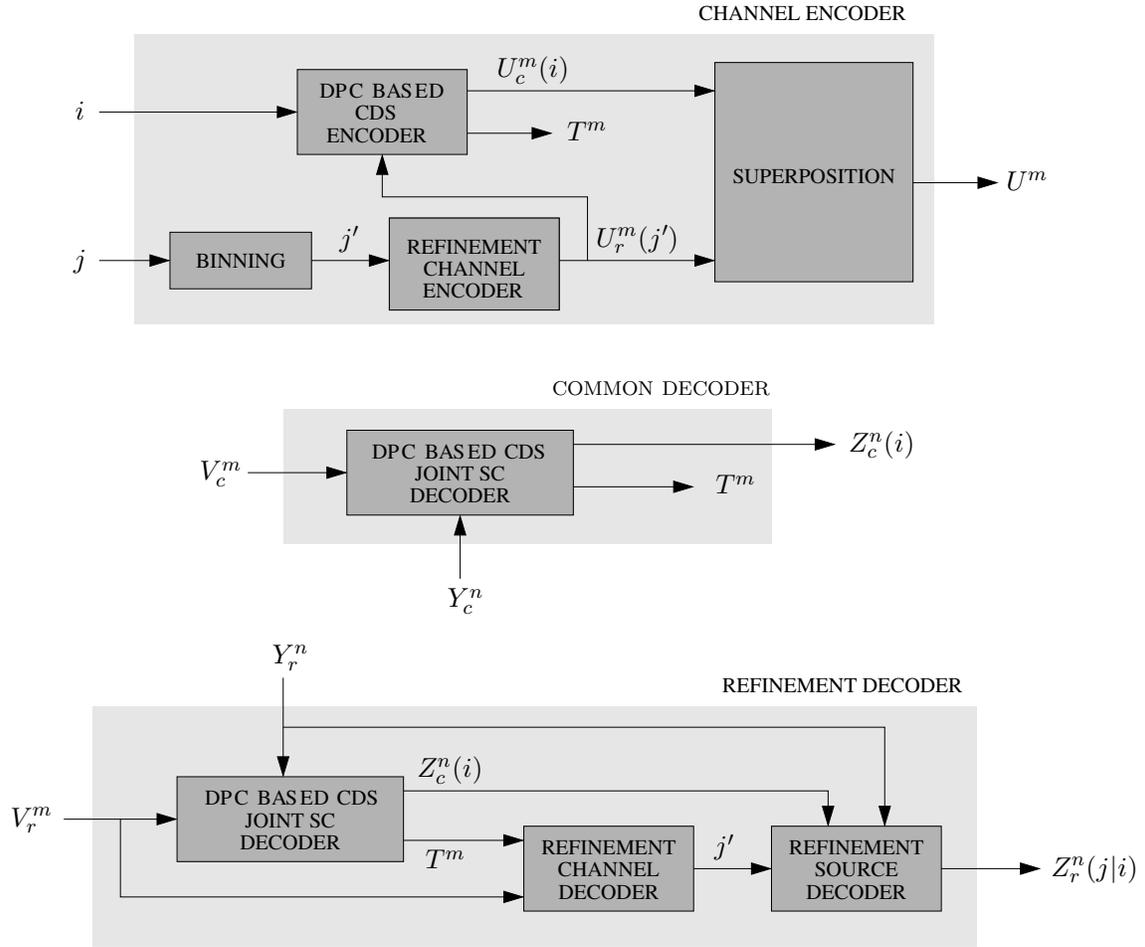}
\caption{Components of LDS: $Z_c^n(i)$ and $Z_r^n(j|i)$ are the first and second stage quantized source words. $Z_r^n(j|i)$ is binned and the bin index $j'$ is channel coded to $U_r^m(j')$ in the usual sense.
$Z_c^n(i)$, on the other hand, is mapped to $U_c^m(i)$ using CDS with DPC, where $U_r^m(j')$ serves as the CSI. The two channel codewords are then superposed, resulting in $U^m$.
Decoding of $Z_c^n(i)$ is exactly as in CDS with DPC at both receivers.
In decoding of $Z_r^n(j|i)$, the refinement channel decoder makes use of both the channel output $V_r^m$ and the auxiliary code word $T^m$ to decode the bin index $j'$.}\label{figr:RCCR}
\end{figure}

%%%%%%%%%%%%%%%%%%%%%%%%%
\subsection{Source Coding Rates for LDS}
\label{subs:SourceCodingAll}
%%%%%%%%%%%%%%%%%%%%%%%%%

The RL is transmitted by separate source and channel coding. In coding the source, we restrict our attention to systems where the communicated information satisfies $(Y_c,Y_r)-X-Z_r-Z_c$ where $Z_c$ corresponds to the CL and $Z_r$ is the RL. The source coding rate for the RL is therefore $I(X;Z_r|Z_c,Y_r)$ (cf.~\cite{TianDiggavi}). This has to be less than the RL capacity. Due to the separability of the source and channel variables in the required inequalities we can say that a distortion pair $(D_c, D_r)$ is achievable if
\[
{\cal R}^s_{\rm LDS}(D_c,D_r) \cap {\cal C}_{\rm WZBC}(\kappa) \neq \emptyset \; .
\]
Here, ${\cal C}_{\rm WZBC}(\kappa)$ is the ``capacity'' region achieved by either LDS or any of its variations discussed in Appendix~\ref{subs:App:OtherSchemes}, and ${\cal R}^s_{\rm LDS}(D_c,D_r)$ is the set of all triplets $(R^s_{cc},R^s_{cr},R^s_{rr})$ so that there exist $(Z_c, Z_r)$ and reconstruction functions $g_c:\mathcal{Z}_c\times \mathcal{Y}_c \to \hat{\mathcal{X}}_c$ and $g_r:\mathcal{Z}_r\times \mathcal{Y}_r \to \hat{\mathcal{X}}_r$ satisfying $(Y_c,Y_r)-X-Z_r-Z_c$ and
\begin{align}
\label{eqtn:RccCcc}
I(X;Z_c|Y_c) & \leq R^s_{cc}\\
\label{eqtn:RcrCcr}
I(X;Z_c|Y_r) &  \leq R^s_{cr} \\
\label{eqtn:RrrCrr}
I(X;Z_r|Z_c,Y_r) & \leq R^s_{rr}\\
\mbox{E}[d_c(X,g_c(Z_c, Y_c))] &\leq D_c \\
\mbox{E}[d_r(X,g_r(Z_r, Y_r))] &\leq D_r \; .
\end{align}
The subscripts $cc$ and $cr$ are used to emphasize transmission of the CL to receivers $c$ and $r$, respectively.
Similarly, the subscript $rr$ refers to transmission of RL to receiver $r$.

%%%%%%%%%%%%%%%%%%%%%%%%%
\subsection{Channel Coding Rates for LDS}
\label{subs:SchemeRCCR}
%%%%%%%%%%%%%%%%%%%%%%%%%
The next theorem provides the effective channel rate region for LDS.
\begin{theorem}
\label{thm:LDS}
Let ${\cal R}^c_{\rm LDS}(\kappa)$ be the union of all $(R^c_{cc},R^c_{cr},R^c_{rr})$ for which there exist $U_c\in{\cal U}_c$, $U_r\in{\cal U}_r$, and $T\in{\cal T}$ with $T-(U_r,U_c)-(V_r, V_c)$ and $(U_r,U_c)-U-(V_r, V_c)$ such that
\begin{align}
\label{eqtn:RCCR1}
R^c_{cc} &\leq \kappa[I(T;V_c) - I(T;U_r)] \\
\label{eqtn:RCCR2}
R^c_{cr} &\leq \kappa[I(T;V_r) - I(T;U_r)] \\
\label{eqtn:RCCR3}
R^c_{rr} &\leq \kappa[I(U_r;T,V_r)] \;.
\end{align}
Then ${\cal R}^c_{\rm LDS}(\kappa)\subseteq{\cal C}_{\rm WZBC}(\kappa)$.
\end{theorem}
\begin{remark}
The various random variables that appear in Theorem \ref{thm:LDS} have the following interpretation: $V_c$ and $V_r$ are the channel outputs when the input is $U$. $U_c$ and $U_r$ correspond to the partial channel codewords that are superposed to form the channel input. Finally $T$ is the auxiliary random variable used in DPC with $U_r$ forming the CSI.
\end{remark}
\begin{remark}
In LDS, a trivial $U_r$ together with $T=U$ reduces to CDS.
\end{remark}
\begin{proof}
We construct an RL codebook with elements from ${\cal T}_{\delta}^m(U_r)$. We then use the CDS with DPC construction with the chosen RL codeword acting as CSI. It follows from Theorem~\ref{thrm:S0_DPC} that the CL information can be successfully decoded (together with the auxiliary codeword $T^m$) at both receivers if \eqref{eqtn:RCCR1} and \eqref{eqtn:RCCR2} are satisfied.
This way, the effective communication system for transmission of RL becomes a channel with $U_r^m$ as input and the pair $T^m$ and $V_r^m$ as output. For reliable transmission, \eqref{eqtn:RCCR3} is then sufficient.
\end{proof}

%%%%%%%%%%%%%%%%%%%%%%%%%%%%%%%%%%%%%%%%%%%%%%%%%%%%%%%%%%%%%%%%%
%%%%%%%%%%%%%%%%%%%%%%%%%%%%%%%%%%%%%%%%%%%%%%%%%%%%%%%%%%%%%%%%%
\section{Performance Analysis for the Quadratic Gaussian Problem}
\label{sctn:QuadraticGaussian}
%%%%%%%%%%%%%%%%%%%%%%%%%%%%%%%%%%%%%%%%%%%%%%%%%%%%%%%%%%%%%%%%%
%%%%%%%%%%%%%%%%%%%%%%%%%%%%%%%%%%%%%%%%%%%%%%%%%%%%%%%%%%%%%%%%%

In this section, we analyze the distortion tradeoff of the LDS for the quadratic Gaussian case.
While CDS with DPC is developed only as a tool to be used in layered WZBC codes, CDS itself is a legitimate WZBC strategy.
We thus analyze its performance in some detail first before proceeding with LDS.
It turns out, somewhat surprisingly, that CDS may in fact be the {\em optimal} strategy for an infinite family of source and channel parameters.
Understanding the performance of CDS also gives insight into which receiver should be chosen as receiver $c$, and which one as receiver $r$. We remind the reader that the variance of a Gaussian random variable $A$ will be denoted by $\mathbf{A}$.

%%%%%%%%%%%%%%%%%%%%%%%%%%%%
\subsection{CDS for the Quadratic Gaussian Problem}
\label{subs:GaussianScheme0}
%%%%%%%%%%%%%%%%%%%%%%%%%%%%

Using the test channel $X=Z+S$ with Gaussian $S$ and $Z$ where $S\perp Z$, and a Gaussian channel input $U$, \eqref{eqtn:Ineq_S0} becomes (cf.~\eqref{eqtn:WZR_Gaussian})
\[
\frac{1}{2}\log \left(1-\vari{N}_k+\frac{\vari{N}_k}{\vari{S}}\right) \leq
\frac{\kappa}{2}\log \left(1+\frac{P}{\vari{W}_k}\right)
\]
for $k=1,\ldots,K$.
In other words,
\[
\frac{1}{\vari{S}} \leq 1 + \min_k \frac{\left(1+\frac{P}{\vari{W}_k}\right)^{\kappa}-1}{\vari{N}_k} \; .
\]
By analyzing \eqref{eqtn:WZD_Gaussian}, it is clear that $\vari{S}$ should be chosen so as to achieve the above inequality with equality.
Substituting that choice in \eqref{eqtn:WZD_Gaussian} yields
\begin{equation}
\label{eqtn:Dk_Scheme0}
\frac{1}{D_k} = \frac{1}{\vari{N}_k}+\min_{k'} \frac{\left(1+\frac{P}{\vari{W}_{k'}}\right)^{\kappa}-1}{\vari{N}_{k'}} \; .
\end{equation}
For all $k^*$ that achieve the minimum in \eqref{eqtn:Dk_Scheme0}, we have
\[
\frac{1}{D_{k^*}} = \frac{\left(1+\frac{P}{\vari{W}_{k^*}}\right)^{\kappa}}{\vari{N}_{k^*}} \; .
\]
Thus, as seen from \eqref{eqtn:TrivialGaussianD}, $D_{k^*}=D_{k^*}^{WZ}(\kappa C_{k^*})$.
This, in particular, means that if
\[
\frac{\left(1+\frac{P}{\vari{W}_k}\right)^{\kappa}-1}{\vari{N}_k}
\]
is a constant, CDS achieves the trivial converse and there is no need for a layered WZBC scheme.
Specialization of \eqref{eqtn:Dk_Scheme0} to the case $\kappa=1$ is also of interest:
\begin{equation}
\label{eqtn:Dk_Scheme0_kappa1}
\frac{1}{D_k} = \frac{1}{\vari{N}_k}+\frac{P}{\max_{k'} \big\{\vari{W}_{k'}\vari{N}_{k'}\big\}} \; .
\end{equation}
In particular, all $k^*$ maximizing $\vari{W}_{k^*}\vari{N}_{k^*}$ achieve $D_{k^*}=D_{k^*}^{WZ}(C_{k^*})$.
Thus, the trivial converse is achieved if $\vari{W}_k\vari{N}_k$ is a constant.

%%%%%%%%%%%%%%%%%%%%%%%%%%%%%%%%%
\subsection{LDS for the Quadratic Gaussian Problem}
\label{subs:GaussianLayeredWZBC}
%%%%%%%%%%%%%%%%%%%%%%%%%%%%%%%%%

For LDS, we begin by analyzing the channel coding performance and then the source coding performance in terms of achievable channel rates.
Then closely examining the channel rate regions, we determine whether $c=1,r=2$, or $c=2,r=1$ is more advantageous given $\kappa$, $P$, $\vari{N}_1$, $\vari{N}_2$, $\vari{W}_1$, and $\vari{W}_2$. The resultant expression when $\kappa=1$ exhibits an interesting phenomenon which we will make use of in deriving closed form expressions for the $(D_c,D_r)$ tradeoff in LDS.

%%%%%%%%%%%%%%%%%%%%%%%%%%%%%%%%%%%%%%%%%%
\subsubsection{Channel Coding Performance}
\label{subs:EffectiveCapacityRegions}

For LDS, we choose channel variables $U_c$ and $U_r$ as independent zero-mean Gaussians with variances $\nu P$ and $\bar{\nu} P$, respectively, with $0\leq \nu\leq 1$, and use the superposition rule $U = U_c + U_r$. Motivated by Costa's construction for the auxiliary random variable $T$, we set $T = \gamma U_r + U_c$.
Using \eqref{eqtn:RCCR1}-\eqref{eqtn:RCCR3}, we obtain achievable $(R^c_{cc},R^c_{cr},R^c_{rr})$ as
\begin{eqnarray}
\label{eqtn:GaussianRCCR1}
R^c_{cc} & = & I(\gamma U_r + U_c;U_c+U_r+W_c) - I(U_r;\gamma U_r + U_c) \nonumber \\
& = & h(U_c+U_r+W_c) + h(U_c) - h(\gamma U_r + U_c,U_c+U_r+W_c)  \nonumber \\
& = & \frac{1}{2}\log\frac{[P+\vari{W}_c] \nu P}{\det\left[\begin{array}{cc}
\gamma^2 \bar{\nu} P +\nu P &\gamma \bar{\nu} P +\nu P \\
\gamma \bar{\nu} P + \nu P & P+\vari{W}_c
\end{array}\right]} \nonumber \\
& = & \frac{1}{2}\log \frac{1+\frac{P}{\vari{W}_c}}{1 + \bar{\nu} P \left(\frac{\gamma^2}{\nu P} + \frac{(1-\gamma)^2}{\vari{W}_c}\right)} \\
\label{eqtn:GaussianRCCR2}
R^c_{cr} & = & I(\gamma U_r + U_c;U_c+U_r+W_r) - I(U_r;\gamma U_r + U_c) \nonumber \\
& = & \frac{1}{2}\log \frac{1+\frac{P}{\vari{W}_r}}{1 + \bar{\nu} P \left(\frac{\gamma^2}{\nu P} + \frac{(1-\gamma)^2}{\vari{W}_r}\right)} \\
\label{eqtn:GaussianRCCR3}
R^c_{rr} & = & I(U_r;\gamma U_r + U_c,U_c+U_r+W_r) \nonumber \\
& = & h(\gamma U_r + U_c,U_c+U_r+W_r) - h(U_c,U_c+W_r) \nonumber \\
& = & h(\gamma U_r + U_c,U_c+U_r+W_r) - h(U_c) - h(W_r) \nonumber \\
& = & \frac{1}{2}\log \frac{\det\left[\begin{array}{cc}
\gamma^2 \bar{\nu} P +\nu P &\gamma \bar{\nu} P +\nu P \\
\gamma \bar{\nu} P + \nu P & P+\vari{W}_r
\end{array}\right]}{\nu P \vari{W}_r}  \nonumber \\
& = & \frac{1}{2}\log \left(1 + \bar{\nu} P \left(\frac{\gamma^2}{\nu P} + \frac{(1-\gamma)^2}{\vari{W}_r}\right) \right) \; .
\end{eqnarray}
Here, \eqref{eqtn:GaussianRCCR2} follows by replacing $W_c$ with $W_r$ in \eqref{eqtn:GaussianRCCR1}.

%%%%%%%%%%%%%%%%%%%%%%%%%%%%%%%%%%%%%%%%%
\subsubsection{Source Coding Performance}

We choose the auxiliary random variables so that $X= Z_r + S_r$ and $Z_r = Z_c + S_c'$ where $S_r$ and $S_c'$ are Gaussian random variables satisfying $S_r\perp Z_r$ and $S'_c\perp Z_c$.
This choice imposes the Markov chain $X-Z_r-Z_c$, and implies $X = Z_c + S_c$ with $S_c \perp Z_c$ and $1 \geq \vari{S}_c \geq \vari{S}_r$.
Using \eqref{eqtn:WZR_Gaussian}, one can then conclude
\begin{eqnarray}
\label{eqtn:GaussianRcc}
R^s_{cc} & = & \frac{1}{2}\log \left( 1-\vari{N}_c + \frac{\vari{N}_c}{\vari{S}_c}\right) \\
\label{eqtn:GaussianRcr}
R^s_{cr} & = & \frac{1}{2}\log \left( 1-\vari{N}_r + \frac{\vari{N}_r}{\vari{S}_c} \right) \\
\label{eqtn:GaussianRrr}
R^s_{rr} & = & \frac{1}{2}\log \left( \frac{1-\vari{N}_r + \frac{\vari{N}_r}{\vari{S}_r}}{1-\vari{N}_r + \frac{\vari{N}_r}{\vari{S}_c}} \right) \; .
\end{eqnarray}
For any achievable triplet $(R^c_{cc},R^c_{cr},R^c_{rr})$, \eqref{eqtn:GaussianRcc}-\eqref{eqtn:GaussianRrr} can be used to find the corresponding best $(D_c,D_r)$.
More specifically, \eqref{eqtn:GaussianRcc}-\eqref{eqtn:GaussianRrr} and \eqref{eqtn:RccCcc}-\eqref{eqtn:RrrCrr} together imply
\begin{eqnarray}
\label{eqtn:varSc}
\frac{1}{\vari{S}_c} & \leq & \min \left\{\frac{2^{2\kappa R^c_{cc}} - 1}{\vari{N}_c}, \frac{2^{2\kappa R^c_{cr}} - 1}{\vari{N}_r}\right\} + 1 \\
\label{eqtn:varSr}
\frac{1}{\vari{S}_r} & \leq & \frac{2^{2\kappa R^c_{rr}}\left(1-\vari{N}_r + \frac{\vari{N}_r}{\vari{S}_c}\right) -1}{\vari{N}_r} + 1 \; .
\end{eqnarray}
Since we have from \eqref{eqtn:WZD_Gaussian} that
\begin{equation}
\label{eqtn:GaussianDk}
D_k = \frac{\vari{N}_k}{1-\vari{N}_k+\frac{\vari{N}_k}{\vari{S}_k}}
\end{equation}
it is easy to conclude that both \eqref{eqtn:varSc} and \eqref{eqtn:varSr} should be satisfied with equality to obtain the best $(D_c,D_r)$, which becomes
\begin{eqnarray}
\label{eqtn:GaussianDc}
D_c & = & \frac{\vari{N}_c}{1 + \vari{N}_c\phi}\\
\label{eqtn:GaussianDr}
D_r & = & \frac{\vari{N}_r}{1 + \vari{N}_r\phi}2^{-2\kappa R^c_{rr}}
\end{eqnarray}
where
\begin{equation}
\label{eqtn:GaussianPhi}
\phi = \min\left\{\frac{2^{2\kappa R^c_{cc}} - 1}{\vari{N}_c}, \frac{2^{2\kappa R^c_{cr}} - 1}{\vari{N}_r}\right\} \; .
\end{equation}

Now, if
\begin{equation}
\label{eqtn:GaussianPhiCondition}
\frac{2^{2\kappa R^c_{cc}} - 1}{\vari{N}_c} \geq \frac{2^{2\kappa R^c_{cr}} - 1}{\vari{N}_r}
\end{equation}
then $D_r = \vari{N}_r 2^{-2 \kappa(R^c_{cr} + R^c_{rr})}$.
But in the LDS, we have $R^c_{cr}+R^c_{rr} = C_r = \frac{1}{2}\log \left(1+\frac{P}{W_r}\right)$, implying $D_r=D_r^{WZ}(\kappa C_r)$, regardless of the chosen parameters. Moreover, $D_c$ will be minimized when \eqref{eqtn:GaussianPhiCondition} is satisfied with equality.
Thus, it suffices to consider only
\begin{equation}
\label{eqtn:GaussianPhiCondition2}
\frac{2^{2\kappa R^c_{cc}} - 1}{\vari{N}_c} \leq \frac{2^{2\kappa R^c_{cr}} - 1}{\vari{N}_r}
\end{equation}
because equality in \eqref{eqtn:GaussianPhiCondition2} already gives $D_r=D_r^{WZ}(\kappa C_r)$.
We thus have
\begin{eqnarray}
\label{eqtn:GaussianDc2}
D_c & = & \vari{N}_c2^{-2\kappa R^c_{cc}}\\
D_r & = & \frac{\vari{N}_r}{1+\frac{\vari{N}_r}{\vari{N}_c}[2^{2 \kappa R^c_{cc}} - 1]}2^{-2 \kappa R^c_{rr}} \nonumber \\
\label{eqtn:GaussianDr2}
& = & \frac{\vari{N}_r}{1+\vari{N}_r\left[\frac{1}{D_c}-\frac{1}{\vari{N}_c}\right]}2^{-2 \kappa R^c_{rr}}\; .
\end{eqnarray}

%%%%%%%%%%%%%%%%%%%%%%%%%%%%%%%%%%%%%%%%%%%%%%%%
\subsubsection{Choosing the Refinement Receiver}

Note that setting $\nu =1$ reduces LDS to CDS.
This is regardless of which receiver is designated as $c$ or $r$.
This simple observation, along with the discussion in Section~\ref{subs:GaussianScheme0}, leads to the following lemma.
\begin{lemma}
In order to maximize the performance of LDS, one must set $c$ and $r$ so that
\begin{equation}
\label{eqtn:LemmaWN}
\frac{\left(1+\frac{P}{\vari{W}_c}\right)^{\kappa}-1}{\vari{N}_c} \leq \frac{\left(1+\frac{P}{\vari{W}_r}\right)^{\kappa}-1}{\vari{N}_r} \; .
\end{equation}
\end{lemma}
\begin{remark}
When $\kappa=1$, \eqref{eqtn:LemmaWN} translates to
\begin{equation}
\label{eqtn:GoldenRule}
\vari{W}_c\vari{N}_c \geq \vari{W}_r\vari{N}_r \; .
\end{equation}
Therefore, the product $\vari{W}_k\vari{N}_k$ determines the {\em combined} channel and side information quality, so that the ``better'' receiver is chosen to receive the RL information. Recall from the discussion in Section~\ref{subs:GaussianScheme0} that if $\vari{W}_k\vari{N}_k$ is constant, then in fact there is no need for refinement, as CDS already achieves the optimal performance.
\end{remark}
\begin{proof}
When $\nu=1$, i.e., when all the power is allocated to the CL, LDS achieves the same performance as CDS. In particular, it achieves the channel rate point
\begin{eqnarray*}
R^c_{cc} & = & C_c \;\; = \;\frac{1}{2}\log \left(1 + \frac{P}{\vari{W}_c}\right) \\
R^c_{cr} & = & C_r \;\; = \;\frac{1}{2}\log \left(1 + \frac{P}{\vari{W}_r}\right) \\
R^c_{rr} & = & 0 \; .
\end{eqnarray*}
If \eqref{eqtn:LemmaWN} does not hold, then from \eqref{eqtn:Dk_Scheme0}, it follows that LDS also achieves $D_r = D^{WZ}_r(\kappa C_r)$ and some $D_c > D^{WZ}_c(\kappa C_c)$.
Now, if we set $\nu<1$, it is obvious that $D_r$ cannot be lowered any further.
We claim that $D_c$ cannot be lowered either.
Therefore, LDS would not be able to achieve a better $(D_c,D_r)$ than what CDS achieves.
On the other hand, sending the refinement to receiver $c$ could potentially result in a better performance.

Towards proving the above claim, observe from \eqref{eqtn:GaussianPhi} that it suffices to show that neither $R^c_{cc}$ nor $R^c_{cr}$ can increase when $\nu<1$ compared to the case $\nu = 1$.
That, in turn, follows by closely examining the expressions for $R^c_{cc}$ and $R^c_{cr}$ in Section~\ref{subs:EffectiveCapacityRegions}.
In particular, for LDS, both \eqref{eqtn:GaussianRCCR1} and \eqref{eqtn:GaussianRCCR2} will be maximized by their corresponding optimal Costa parameters, i.e., by $\gamma = \frac{\nu P}{\nu P+\vari{W}_c}$ and by $\gamma = \frac{\nu P}{\nu P+\vari{W}_r}$, respectively.
This results in $R^c_{cc} = \frac{1}{2}\log\left(1 + \frac{\nu P}{\vari{W}_c}\right)$ and $R^c_{cr} = \frac{1}{2}\log\left(1 + \frac{\nu P}{\vari{W}_r}\right)$ as the maximum possible values, which are strictly smaller than $C_c$ and $C_r$, respectively.
Therefore, the proof is complete.
\end{proof}

%%%%%%%%%%%%%%%%%%%%%%%%%%%%%%%%%%%%%%%%%%%%%%%%%%%
\subsection{Performance Comparisons for the Bandwidth Matched Case: $\kappa=1$}
%%%%%%%%%%%%%%%%%%%%%%%%%%%%%%%%%%%%%%%%%%%%%%%%%%%

We first derive the closed-form $(D_c,D_r)$ tradeoff for LDS.

\begin{lemma}
\label{lmma:ClosedFormRCCR}
A distortion pair $(D_c,D_r)$ is achievable using LDS if and only if $D_r\geq D_{\rm LDS}(D_c)$, where $D_{\rm LDS}(D_c)$ is the convex hull of
\begin{equation}
\label{eqtn:GaussianRCCR_Dr}
D_{\rm LDS}^*(D_c) = \frac{\vari{N}_r\vari{N}_c^2}{D_c\vari{N}_c
+\vari{N}_r (\vari{N}_c-D_c)}
\cdot \left\{\begin{array}{ll}
\frac{\vari{W}_r D_c}{(\vari{W}_r-\vari{W}_c)\vari{N}_c+(P + \vari{W}_c)D_c} & \vari{W}_c>\vari{W}_r \\
\frac{\vari{W}_c}{P+\vari{W}_c} & \vari{W}_c\leq\vari{W}_r
\end{array}\right.
\end{equation}
for
\[
\frac{\vari{N}_c\vari{W}_c}{P+\vari{W}_c} \leq D_c \leq D_c^{\max}
\]
with
\begin{equation}
\label{eqtn:GaussianRCCR_Dcmax}
D_c^{\max} = \vari{N}_c \cdot\left\{\begin{array}{ll}
\min\left\{1,\frac{\vari{N}_r(\vari{W}_c-\vari{W}_r)}{(P+\vari{W}_c)(\vari{N}_r-\vari{N}_c)}\right\} &
\vari{N}_c<\vari{N}_r , \vari{W}_c>\vari{W}_r \\
1 & \vari{N}_c \geq \vari{N}_r , \vari{W}_c\geq\vari{W}_r \\
\frac{\vari{W}_c}{P+\vari{W}_c} + \frac{P(\vari{W}_c\vari{N}_c-\vari{W}_r\vari{N}_r)}{(P+\vari{W}_c)(\vari{N}_c-\vari{N}_r)\vari{W}_r}  & \vari{N}_c > \vari{N}_r, \vari{W}_c<\vari{W}_r
\end{array}\right. \; .
\end{equation}
\end{lemma}
\begin{remark}
The cases $\vari{N}_c \leq \vari{N}_r , \vari{W}_c<\vari{W}_r $ and $\vari{N}_c < \vari{N}_r , \vari{W}_c=\vari{W}_r$ are not considered in \eqref{eqtn:GaussianRCCR_Dcmax} because they are prohibited by the rule \eqref{eqtn:GoldenRule}.
The same rule also guarantees $\frac{\vari{N}_c\vari{W}_c}{P+\vari{W}_c}\leq D_c^{\max}\leq\vari{N}_c$.
\end{remark}

As a byproduct of the proof, which is deferred to Appendix~\ref{subs:App_lmma:ClosedFormRCCR}, we observe that the Costa parameter $\gamma$ is either $0$ or $1$, depending on whether $\vari{W}_c>\vari{W}_r$ or $\vari{W}_c\leq\vari{W}_r$, respectively. When it is $0$, we have $T=U_c$. On the other hand, when $\gamma=1$, we have $T=U=U_c+U_r$.
Thus, setting the auxiliary codeword $T^m$ to be the same as the channel input $U^m$ constitutes the optimal choice.
To the best of our knowledge, this choice, which is typically encountered in DPC for binary symmetric channels, has never been obtained as the optimal choice involving Gaussian channels.

\begin{figure}
\centering
\subfigure[]{
\includegraphics[width=0.46\columnwidth]{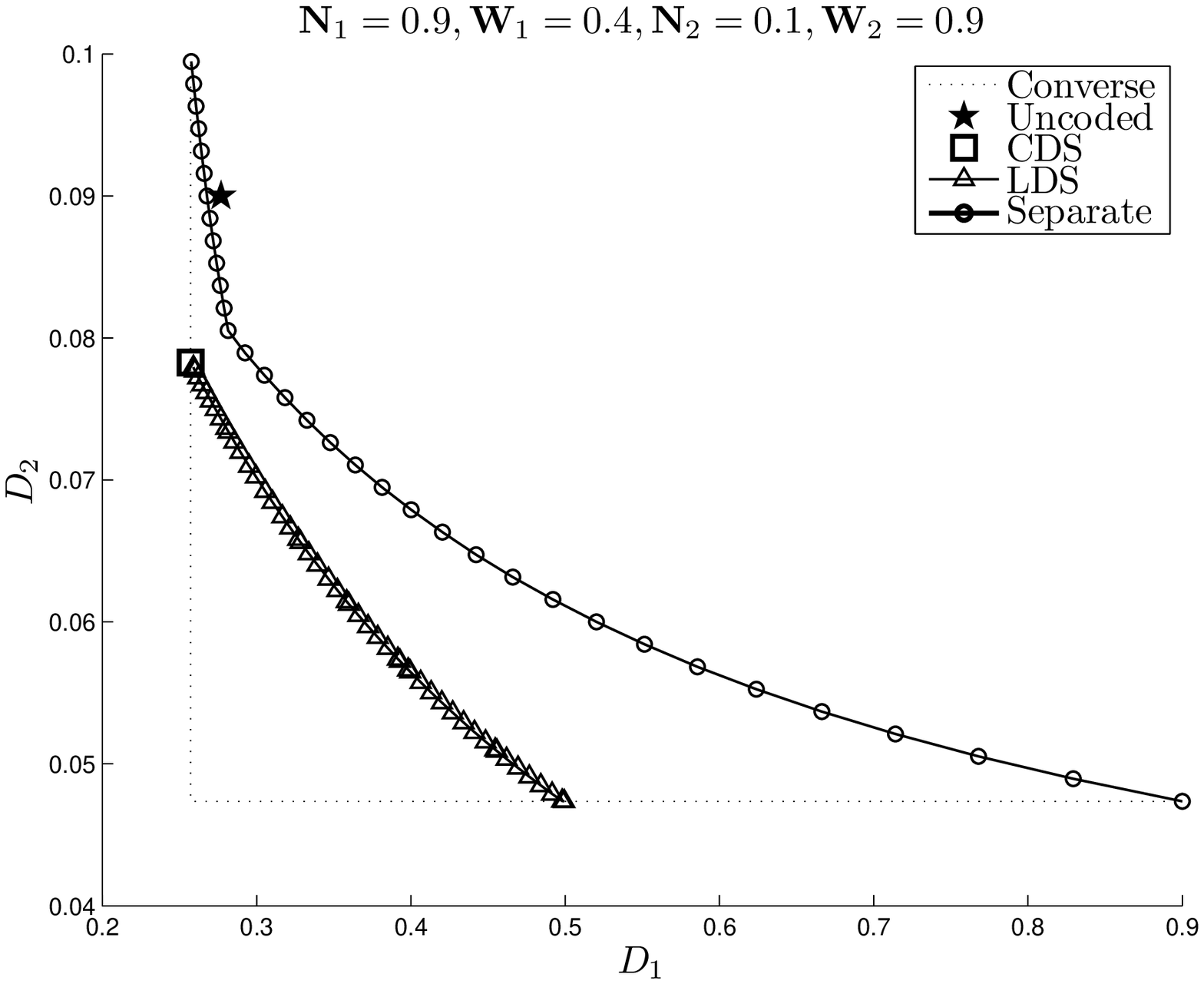}
}
\subfigure[]{
\includegraphics[width=0.46\columnwidth]{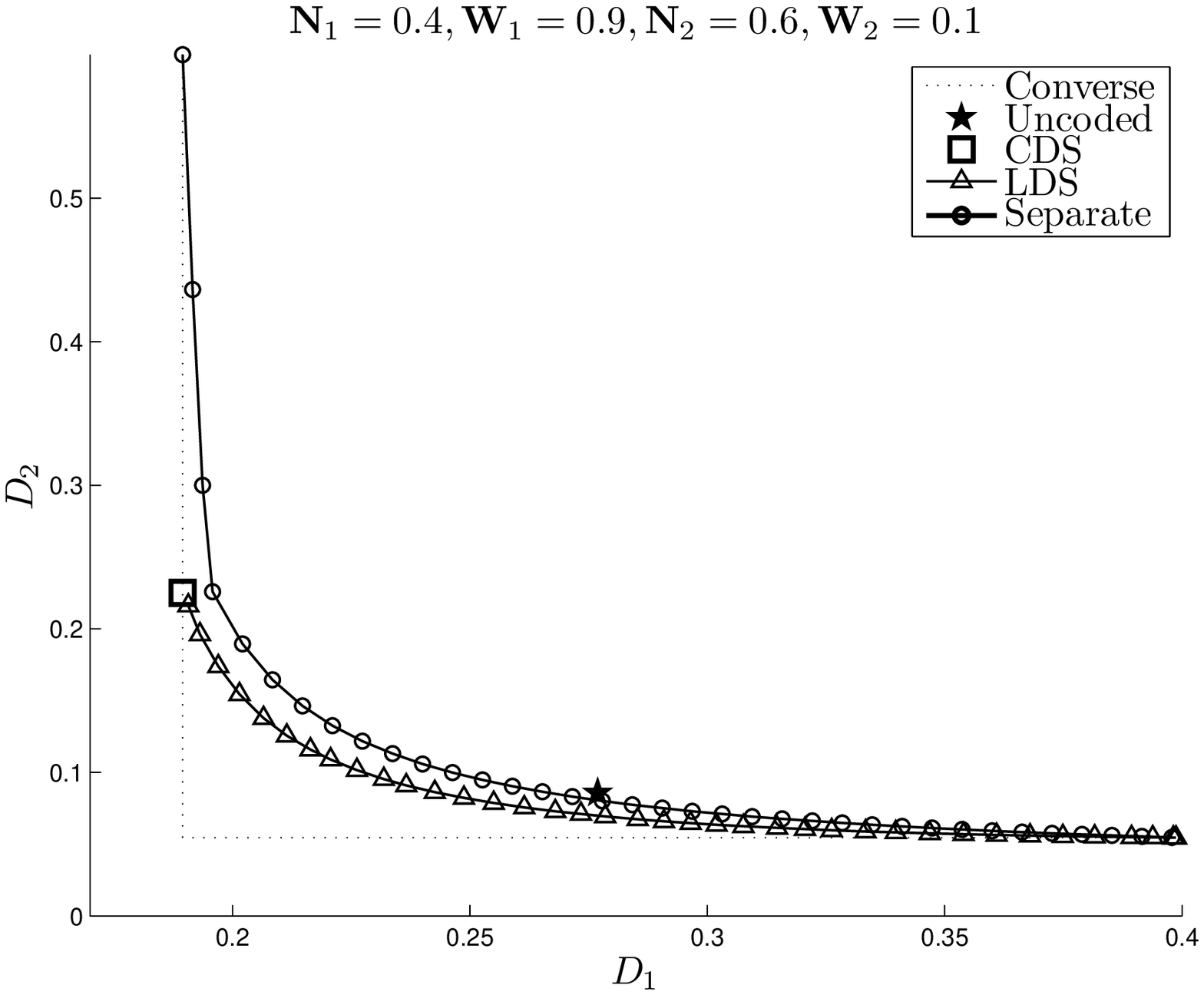}
}
\subfigure[]{
\includegraphics[width=0.46\columnwidth]{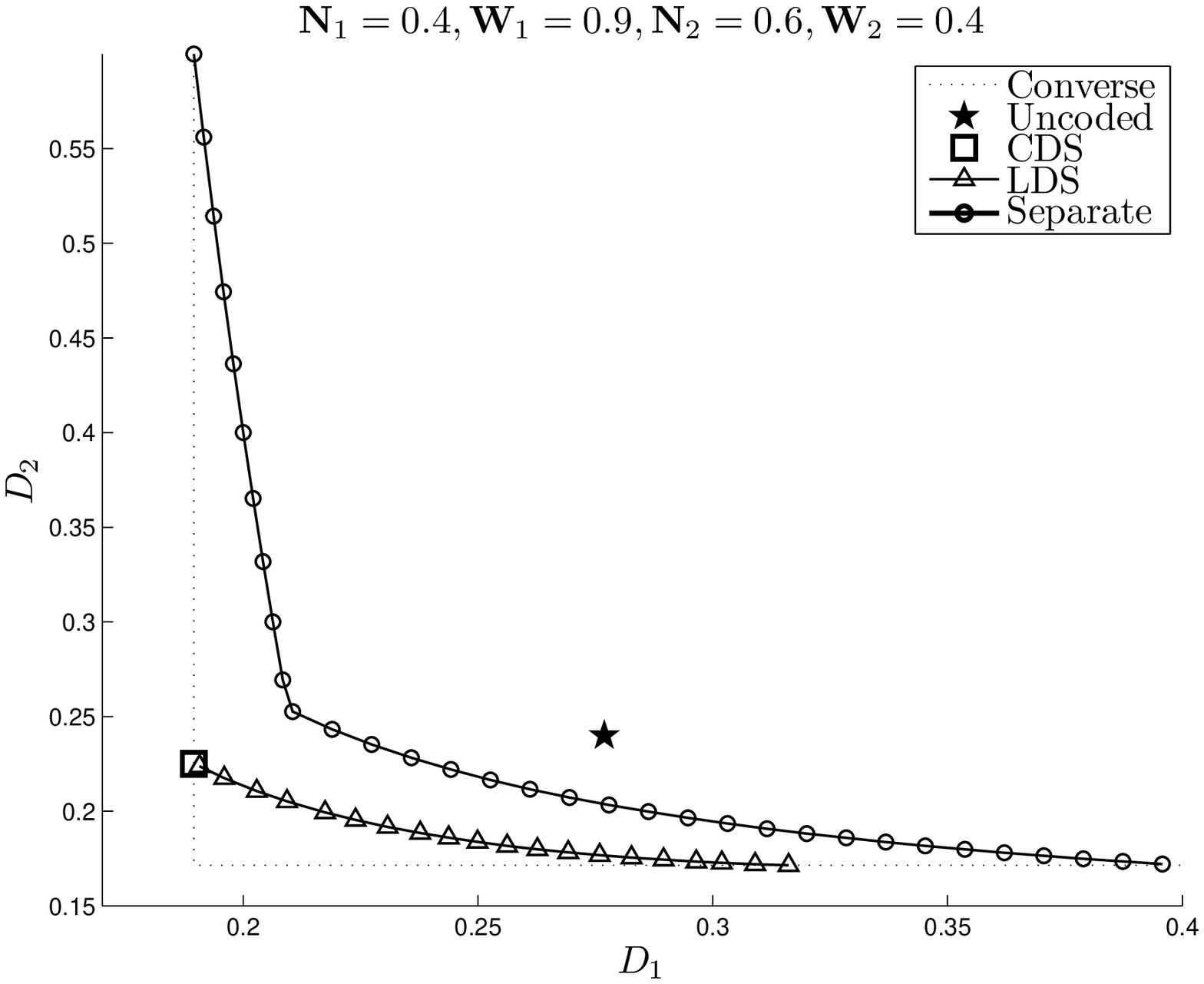}
}
\subfigure[]{
\includegraphics[width=0.46\columnwidth]{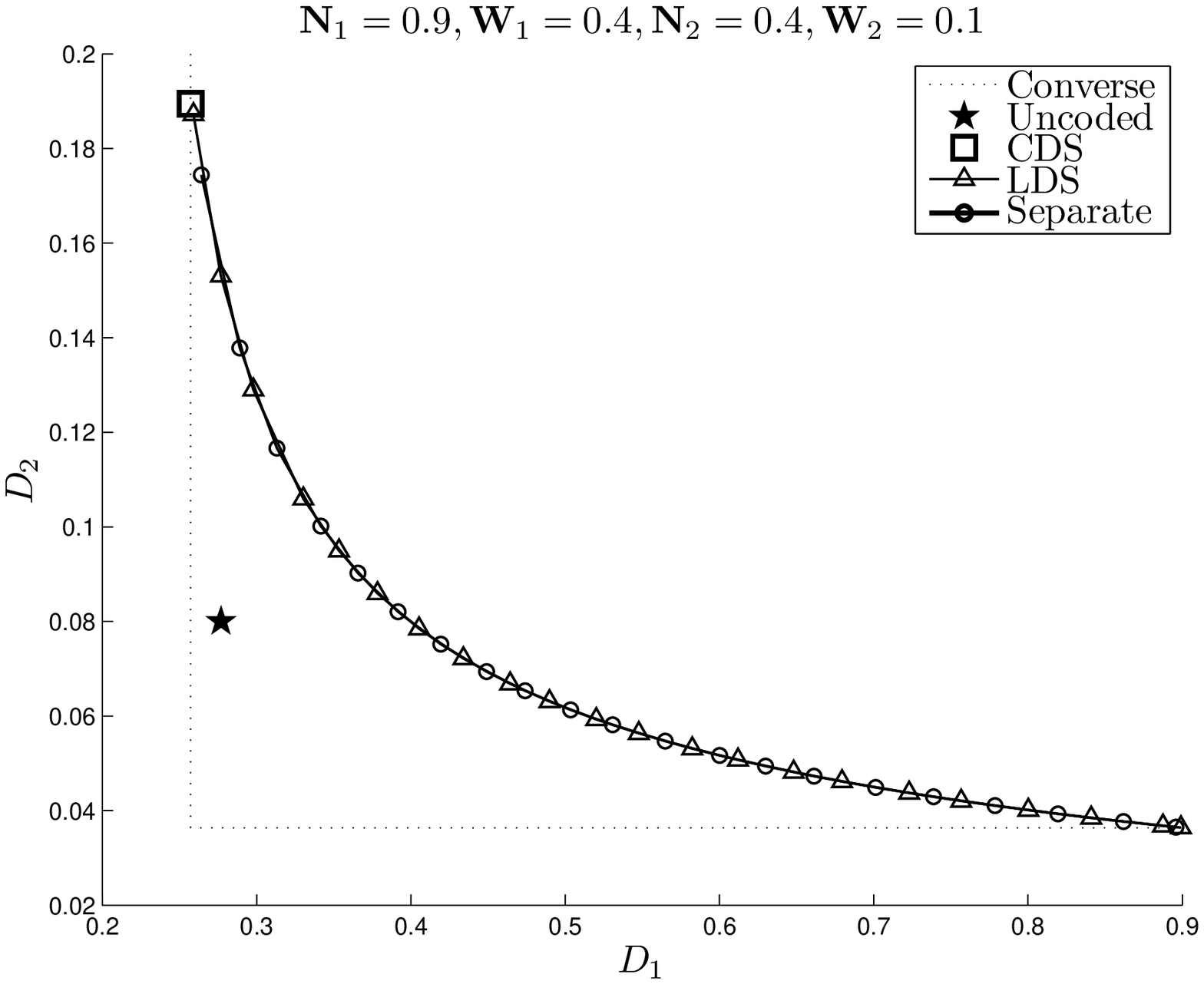}
}
\subfigure[]{
\includegraphics[width=0.46\columnwidth]{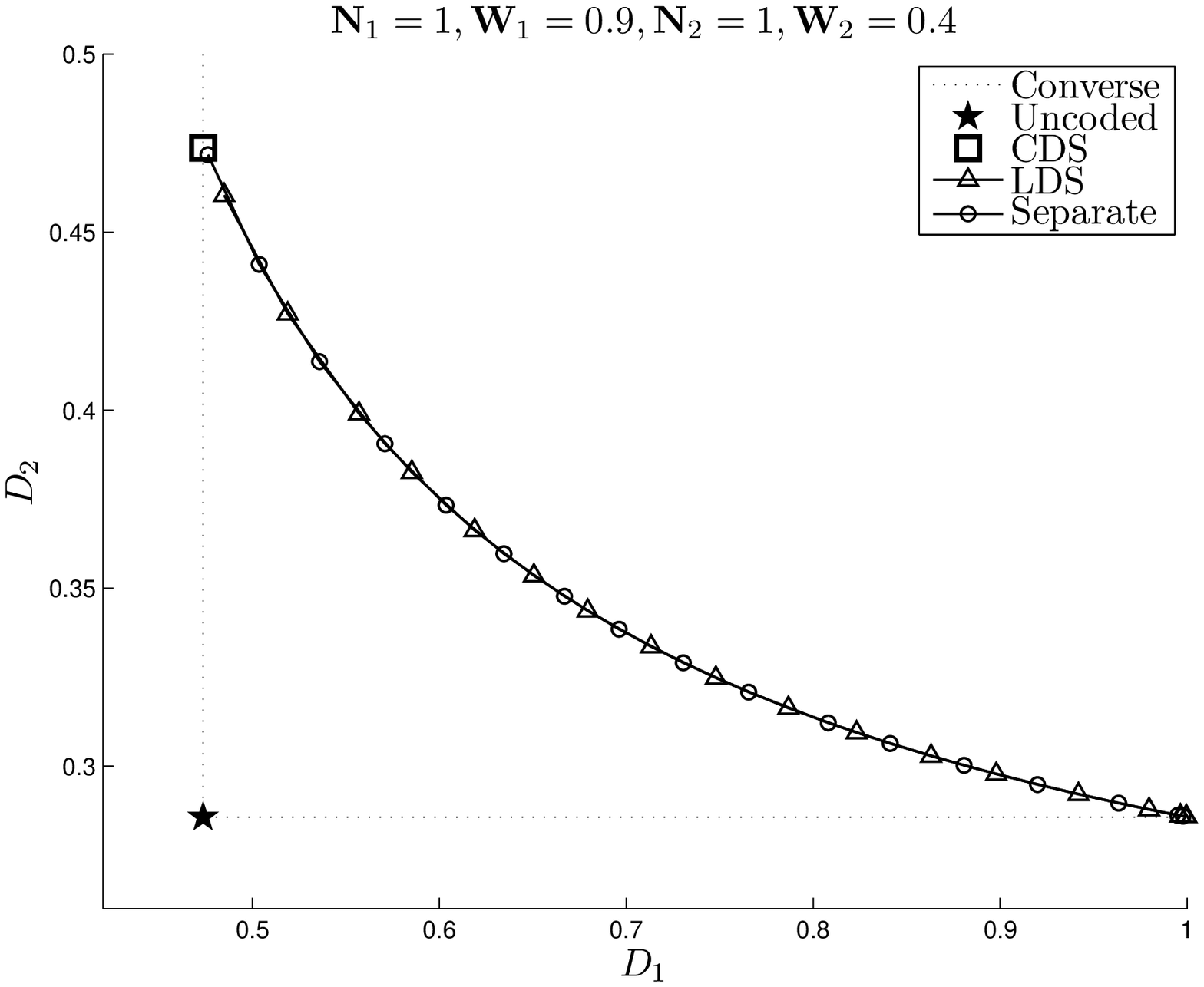}
}
\subfigure[]{
\includegraphics[width=0.46\columnwidth]{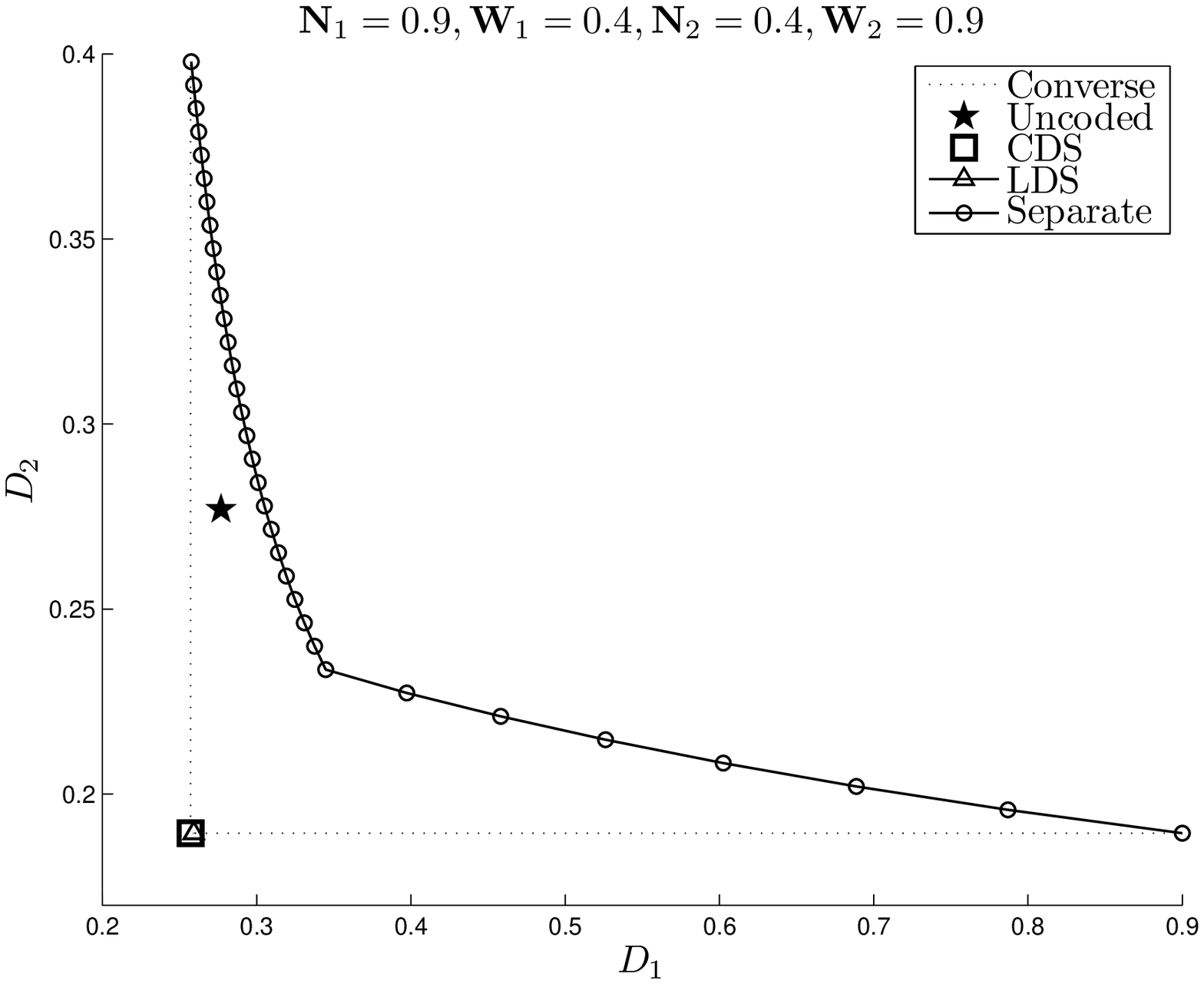}
}
\caption{Performance comparison for Gaussian sources and channels. In (a)-(e), $\vari{N}_1\vari{W}_1 > \vari{N}_2\vari{W}_2$, and therefore the choice $c=1$, $r=2$ is made.
In addition, in (e), $\vari{N}_1=\vari{N}_2=1$, implying that there is no side information at either receiver and hence uncoded transmission is optimal.
In (f), $\vari{N}_1\vari{W}_1 = \vari{N}_2\vari{W}_2$ making CDS optimal.}\label{figr:Gaus}
\end{figure}

We now compare LDS with other schemes for the WZBC problem. The performance of uncoded transmission is governed by \eqref{eqtn:WZ_UncodedGaussian}. The distortion trade-off of separate coding is given by the following lemma, which is proved in Appendix~\ref{subs:App_lmma:GaussianSeparate}. Recall that the subscripts $b$ and $g$ refer to good and bad channels, i.e., the Markov chain $U-V_g-V_b$ holds for all channel inputs $U$.

\begin{lemma}
\label{lmma:GaussianSeparate}
For the quadratic Gaussian case with $\kappa=1$, the distortion pair $(D_b,D_g)$ with $D_b^{WZ}(C_b)\leq D_b\leq\vari{N}_b$ is achievable using separate coding if and only if $D_g\geq D_{\rm SEP}(D_b)$ where $D_{\rm SEP}(D_b)$ is the convex hull of
\begin{equation}
\label{eqtn:GaussianSeparateBetterSideBetterChannel2}
D_{\rm SEP}^*(D_b)  = \frac{\vari{N}_g\vari{N}_b^2 \vari{W}_g D_b}{\Big(D_b\vari{N}_b
+\vari{N}_g (\vari{N}_b-D_b)\Big)\Big((\vari{W}_g-\vari{W}_b)\vari{N}_b+(P + \vari{W}_b)D_b \Big) }
\end{equation}
when $X-Y_g-Y_b$, and
\begin{equation}
\label{eqtn:GaussianSeparateBetterSideWorseChannel2}
D_{\rm SEP}^*(D_b) = \frac{\vari{N}_g}{\Big((\vari{W}_g-\vari{W}_b)\vari{N}_b+(P + \vari{W}_b)D_b \Big)}
\max\Bigg\{\vari{W}_gD_b, \frac{\vari{N}_b\Big(\vari{N}_g\vari{W}_g-(P+\vari{W}_b)D_b-\vari{N}_b(\vari{W}_g-\vari{W}_b)\Big)}{\vari{N}_g-\vari{N}_b}\Bigg\}
\end{equation}
when $X-Y_b-Y_g$.
\end{lemma}

The relative performance of the various schemes will be discussed case by case.
\begin{enumerate}
\item
It is obvious by comparing \eqref{eqtn:GaussianSeparateBetterSideBetterChannel2} and \eqref{eqtn:GaussianRCCR_Dr} that
when $\vari{W}_c\geq\vari{W}_r$ and $\vari{N}_c\geq\vari{N}_r$, LDS obtains the exact same performance as in separate source and channel coding (Note that $r=g,c=b$ in this case). The case where there is no side information, i.e., $\vari{N}_1=\vari{N}_2=1$, falls under this category since the refinement information must go the receiver with the better channel. Therefore we see that the purely digital LDS is worse than the schemes analyzed in~\cite{Reznic} in the absence of side information. Preliminary results from combining LDS with hybrid analog/digital schemes as in~\cite{Reznic} were presented in~\cite{Deniz}.
%This was expected because in this case, Scheme RC-CR is identical to Scheme CR-CR, and according to \eqref{eqtn:CRCR1} and \eqref{eqtn:CRCR2}, the quality of the CL information is limited by the quality of the worse receiver, as in separate coding.
This behavior is displayed in Figures~\ref{figr:Gaus}(d) and (e).

As for uncoded transmission, it can be better than the digital schemes.
For example, consider the case $\vari{N}_c=\vari{N}_r=1$ depicted in Figure~\ref{figr:Gaus}(e), which corresponds to no side information at the receivers. In this case, uncoded transmission actually achieves the trivial converse, and therefore, is the optimal strategy.

\item
When $\vari{W}_c>\vari{W}_r$ and $\vari{N}_c<\vari{N}_r$, it follows from  \eqref{eqtn:GaussianSeparateBetterSideWorseChannel2} and \eqref{eqtn:GaussianRCCR_Dr} that a sufficient condition for superiority of LDS over separate coding is given by
\[
\frac{\vari{N}_r\vari{W}_rD_c}{(\vari{W}_r-\vari{W}_c)\vari{N}_c+(P + \vari{W}_c)D_c} \geq \frac{\vari{N}_r\vari{N}_c^2\vari{W}_r D_c}{\Big(D_c\vari{N}_c +\vari{N}_r (\vari{N}_c-D_c)\Big)\Big((\vari{W}_r-\vari{W}_c)\vari{N}_c+(P + \vari{W}_c)D_c\Big)}
\]
which simplifies to
\[
1 \geq \frac{\vari{N}_c^2}{D_c\vari{N}_c +\vari{N}_r (\vari{N}_c-D_c)}
\]
and is therefore granted since $\vari{N}_c<\vari{N}_r$. Moreover, equality is satisfied, i.e., the two schemes have equal performance, only when $D_c=D_c^{\max}=\vari{N}_c$.
This behavior is exemplified in Figures~\ref{figr:Gaus}(b) and (c).
The difference between the two examples is that $D_c^{\max}=\vari{N}_c$ in (b), whereas $D_c^{\max}<\vari{N}_c$ in (c).

Even though $\vari{N}_c=\vari{N}_r=1$ is prohibited in this case, one can consider $\vari{N}_c=1-\epsilon$ and $\vari{N}_r=1$ with arbitrarily small $\epsilon>0$. Uncoded transmission is also superior to all the digital schemes in this limiting case.

\item
Finally, when $\vari{W}_c<\vari{W}_r$ and $\vari{N}_c>\vari{N}_r$, since $r=b,c=g$ in this case, we need to explicitly write the best $D_c$ for a given $D_r$ for LDS.
From \eqref{eqtn:GaussianRCCR_Dr}, it follows that LDS can achieve
\begin{equation}
\label{eqtn:ComparisonRCCRSeparate1}
D_c = \frac{\vari{N}_c\vari{N}_r}{\vari{N}_c-\vari{N}_r
}\left[\frac{\vari{N}_c\vari{W}_c}{(P+\vari{W}_c)D_r}-1\right]
\end{equation}
for $D_r^{WZ}(C_r)\leq D_r\leq \frac{\vari{N}_c\vari{N}_r\vari{W}_c}{\vari{N}_c\vari{W}_c + P\vari{N}_r}$.
On the other hand, \eqref{eqtn:GaussianSeparateBetterSideWorseChannel2} implies that the minimum $D_c$ that can be achieved by separate coding must necessarily satisfy
\begin{eqnarray}
D_c & \geq & \frac{\vari{N}_c\vari{N}_r\Big(\vari{N}_c\vari{W}_c-(P+\vari{W}_r)D_r-\vari{N}_r(\vari{W}_c-\vari{W}_r)\Big)}{\Big((\vari{W}_c-\vari{W}_r)\vari{N}_r+(P + \vari{W}_r)D_r \Big)(\vari{N}_c-\vari{N}_r)}\nonumber \\
\label{eqtn:ComparisonRCCRSeparate2}
& = & \frac{\vari{N}_c\vari{N}_r}{\vari{N}_c-\vari{N}_r
}\left[\frac{\vari{N}_c\vari{W}_c}{(\vari{W}_c-\vari{W}_r)\vari{N}_r+(P+\vari{W}_c)D_r}-1\right] \; .
\end{eqnarray}
Superiority of LDS over separate coding then easily follows from \eqref{eqtn:ComparisonRCCRSeparate1} and \eqref{eqtn:ComparisonRCCRSeparate2}.
An example of this case is shown in Figure~\ref{figr:Gaus}(a).

We next show that LDS always outperforms uncoded transmission in this case.
In fact, uncoded transmission is even worse than CDS.
Since CDS achieves $D_r=D_r^{WZ}(C_r)$, it suffices to compare the $D_r$ values.
Comparing \eqref{eqtn:WZ_UncodedGaussian} and \eqref{eqtn:Dk_Scheme0_kappa1}, this reduces to showing
\[
\frac{\vari{N}_r\vari{W}_r}{\vari{W}_r + \vari{N}_r P} \geq \frac{\vari{N}_r\vari{N}_c\vari{W}_c}{\vari{N}_c\vari{W}_c + \vari{N}_r P}
\]
or equivalently
\[
\vari{W}_r \geq \vari{N}_c\vari{W}_c \; .
\]
But since $\vari{W}_r >\vari{W}_c$, this is trivially true.
\end{enumerate}

In Figure~\ref{figr:Gaus}(f), we also include an example where $\vari{N}_c\vari{W}_c = \vari{N}_r\vari{W}_r$, i.e., where the combined channel and side information qualities are the same.
CDS achieves the trivial converse as discussed in Section~\ref{subs:GaussianScheme0}.
We also observed that uncoded transmission may achieve a distortion pair below the best known digital tradeoff, as shown in Figures~\ref{figr:Gaus}(d) and (e).
This was expected because it is well-known that the optimal scheme is uncoded transmission when there is no side information at either receiver, as is the case in Figure~\ref{figr:Gaus}(e).
For cases other than $\vari{N}_c\vari{W}_c = \vari{N}_r\vari{W}_r$, one could roughly say that LDS is better than uncoded transmission when the quality of the side information is sufficiently high, although we do not currently have the analytical means for comparison.

%The flexibility in the choice of $\gamma$ is a feature that is not present in other scenarios where DPC has been used for coding over broadcast channels~\cite{Caire}. In those scenarios, there is a unique value of $\gamma$, as specified in~\cite{Costa}, that is optimal. In the layered WZBC schemes, $\gamma$ can be used to tradeoff between $(R^c_{cc}, R^c_{cr})$ and $R^c_{rr}$.

%%%%%%%%%%%%%%%%%%%%%%%%%%%%%%%%%%%%%%%%%%%%%%%%%%%%%%%%%%%%%%%%%
%%%%%%%%%%%%%%%%%%%%%%%%%%%%%%%%%%%%%%%%%%%%%%%%%%%%%%%%%%%%%%%%%
\section{Performance Analysis for the Binary Hamming Problem}
\label{sctn:BinaryHamming}
%%%%%%%%%%%%%%%%%%%%%%%%%%%%%%%%%%%%%%%%%%%%%%%%%%%%%%%%%%%%%%%%%
%%%%%%%%%%%%%%%%%%%%%%%%%%%%%%%%%%%%%%%%%%%%%%%%%%%%%%%%%%%%%%%%%

In this section, we first analyze the CDS for the binary Hamming problem and show that it can be optimal in this case as well.
We then analyze the LDS and present numerical comparisons of the LDS with separate coding and uncoded transmission.

%%%%%%%%%%%%%%%%%%%%%
\subsection{CDS for the Binary Hamming Problem}
%%%%%%%%%%%%%%%%%%%%%

It follows from Corollary~\ref{corr:S0} and Equations \eqref{eqtn:WZR_Binary} and \eqref{eqtn:WZD_Binary} that in the binary Hamming case, if there exists $0\leq q\leq 1$ and $0\leq\alpha\leq\frac{1}{2}$ such that
\begin{equation}
\label{eqtn:BinaryScheme0Rate}
q r(\alpha,\beta_k) \leq \kappa[1-H_2(p_k)]
\end{equation}
for all $k$, then
\begin{equation}
\label{eqtn:BinaryScheme0Distortion}
D_k = (1-q)\beta_k + q\min\{\alpha,\beta_k\}
\end{equation}
can be achieved by the CDS.
Unlike in the quadratic Gaussian case, the constraint \eqref{eqtn:BinaryScheme0Rate} does not result in a single best value for $q$ and $\alpha$.
Therefore, CDS produces a tradeoff of $D_k$'s rather than one best point.

As discussed at the end of Section~\ref{subs:P2P}, the distortion-rate function $D_k^{WZ}(R)$ is achieved either by $q=1$ and $\alpha\leq\alpha_0(\beta_k)$, or by $0\leq q<1$ and $\alpha=\alpha_0(\beta_k)$.
The implication of this fact to the CDS is the following:
\begin{enumerate}
\item If $\beta_k$ are not identical, neither are $\alpha_0(\beta_k)$, and thus we need $q=1$ and some $\alpha\leq\min_k \alpha_0(\beta_k)$ to attain all $D^{WZ}_k(\kappa C_k)$ simultaneously, i.e.,
\begin{equation}
\label{eqtn:BinaryScheme0optimal}
r(\alpha,\beta_k)  = \kappa[1-H_2(p_k)]
\end{equation}
for all $k$.
When this happens, we must necessarily have $D_k=\alpha$ i.e., $D_k$ does not depend on $k$.
\item If $\beta_k=\beta$ for $k=1,\ldots,K$, and thus $D_k^{WZ}(R)$ does not depend on $k$, we need $C_k=C$ (and hence $p_k=p$) so that the same test channel $(q,\alpha)$ achieves $D_k^{WZ}(C_k)$ simultaneously. But, this makes the problem trivial.
\end{enumerate}

%%%%%%%%%%%%%%%%%%%%%%%%%%%%%%%%%
\subsection{LDS for the Binary Hamming Problem}
\label{subs:BinaryRCCR}
%%%%%%%%%%%%%%%%%%%%%%%%%%%%%%%%%
\begin{figure}
\centering
\psfrag{X}{\large $X$}
\psfrag{Zc}{\large $Z_b$}
\psfrag{Zr}{\large $Z_g$}
\psfrag{0}{0}
\psfrag{1}{1}
\psfrag{lam}{$\lambda$}
\psfrag{q1}{$\bar{q}_g$}
\psfrag{p1}{$\bar{q}'_b$}
\psfrag{a1}{$q_g \bar{\alpha}_g$}
\psfrag{a2}{$q_g \alpha_g$}
\psfrag{b1}{$q'_b \bar{\alpha}'_b$}
\psfrag{b2}{$q'_b \alpha'_b$}
\includegraphics[width=9cm]{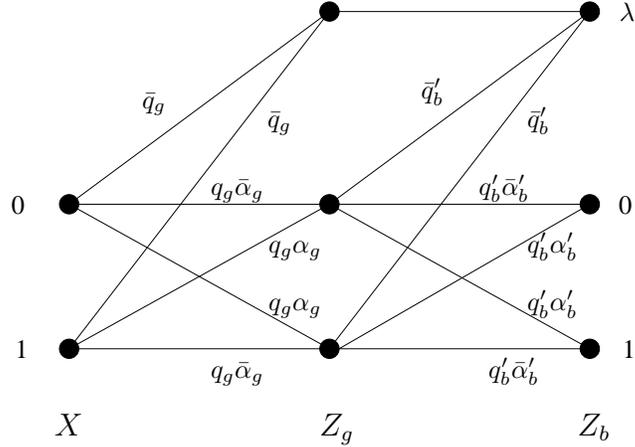}
\caption{Auxiliary random variables for binary source coding. The edge labels denotes transition probabilities. We also use the convention that $\bar{a}=1-a$.}\label{figr:BAux}
\end{figure}

\subsubsection{Source Coding Rates}
To evaluate $R^s_{cc}, R^s_{cr}$ and $R^s_{rr}$, we first fix $Z_c$ and $Z_r$ with ${\cal Z}_c={\cal Z}_r=\{0,1,\lambda\}$, where the test channels are also confined to degraded versions of those that achieve $D^{WZ}(R)$, as shown in Figure~\ref{figr:BAux} for the case $(Y_c,Y_r)-X-Z_r-Z_c$. More specifically,
\begin{eqnarray*}
Z_c & = & E_c\circ(X\oplus S_c) \\
Z_r & = & E_r\circ(X\oplus S_r)
\end{eqnarray*}
where $E_c,E_r,S_c$, and $S_r$ are all Bernoulli random variables with parameters $q_c,q_r,\alpha_c$, and $\alpha_r$, respectively.
To obtain a Markov relation $X-Z_r-Z_c$, it suffices to enforce $q_c\leq q_r$ and $\alpha_c\geq \alpha_r$. In that case, one can find $0\leq q'_c\leq 1$ and $0\leq\alpha'_c\leq\frac{1}{2}$ such that $q_c=q_r q'_c$ and $\alpha_c = \alpha_r\star\alpha'_c$, and $Z_c$ can alternatively be written as
\[
Z_c = \begin{cases} E'_c\circ(Z_r\oplus S'_c) & Z_r \neq \lambda\\
\lambda & Z_r = \lambda
\end{cases}
\]
where $E'_c$ and $S'_c$ are $\Ber(q'_c)$ and $\Ber(\alpha'_c)$, respectively.

This results in
\begin{eqnarray*}
R^s_{cc} & = & q_c r(\alpha_c, \beta_c) \\
R^s_{cr} & = & q_c r(\alpha_c, \beta_r) \\
R^s_{rr} & = & q_r r(\alpha_r, \beta_r) - q_c r(\alpha_c, \beta_r) \; .
\end{eqnarray*}

We next make channel variable choices and derive the resulting channel coding rates for CDS and LDS individually. Unlike in the quadratic Gaussian case, there is no power allocation parameter to vary. However, we have freedom in choosing the distributions of $U_c$ and $U_r$ as $\Ber(\gamma_c)$ and $\Ber(\gamma_r)$, respectively, as well as in choosing the auxiliary random variable as either $T=U_c$ or $T=U_c\oplus U_r$.

%%%%%%%%%%%%%%%%%%%%%%%%%%%%
\subsubsection{Channel Coding Rates}

In this case, with $T=U_c$, \eqref{eqtn:CRRC1}-\eqref{eqtn:CRRC3} become
\begin{eqnarray}
R^c_{cc} & = & I(U_c; U_c\oplus U_r \oplus W_c) \nonumber \\
& = & r(\gamma_r \star p_c,\gamma_c) \nonumber \\
R^c_{cr} & = & I(U_c; U_c\oplus U_r \oplus W_r) \nonumber \\
& = & r(\gamma_r \star p_r,\gamma_c) \nonumber \\
R^c_{rr} & = & I(U_c\oplus U_r; U_c\oplus U_r \oplus W_r |U_c ) \nonumber \\
& = & I(U_r; U_r \oplus W_r ) \nonumber \\
\label{eqtn:BinaryCRCR3}
& = & r(p_r,\gamma_r) \; .
\end{eqnarray}
But since $r(\cdot,\cdot)$ is increasing in its second argument, we have $\gamma_c=\frac{1}{2}$ as the optimal value achieving
\begin{eqnarray}
\label{eqtn:BinaryCRCR1}
R^c_{cc} & = & 1-H_2(\gamma_r \star p_c) \\
\label{eqtn:BinaryCRCR2}
R^c_{cr} & = & 1-H_2(\gamma_r \star p_r) \; .
\end{eqnarray}

On the other hand, if $T=U_c\oplus U_r$, we obtain
\begin{eqnarray}
R^c_{cc} & = & I(U_c\oplus U_r; U_c\oplus U_r \oplus W_c) - I(U_r;U_c\oplus U_r)\nonumber \\
\label{eqtn:BinaryRCCR1}
& = & r(p_c,\gamma_c\star\gamma_r) - r(\gamma_c,\gamma_r)  \\
R^c_{cr} & = & I(U_c\oplus U_r; U_c\oplus U_r \oplus W_r) - I(U_r;U_c\oplus U_r)\nonumber \\
\label{eqtn:BinaryRCCR2}
& = & r(p_r,\gamma_c\star\gamma_r) - r(\gamma_c,\gamma_r) \\
R^c_{rr} & = & I(U_r; U_c\oplus U_r, U_c\oplus U_r \oplus W_r) \nonumber \\
& = & I(U_r; U_c\oplus U_r) \nonumber \\
\label{eqtn:BinaryRCCR3}
& = &  r(\gamma_c,\gamma_r)\; .
\end{eqnarray}

\begin{figure}
\centering
\subfigure[]{
\includegraphics[width=0.46\columnwidth]{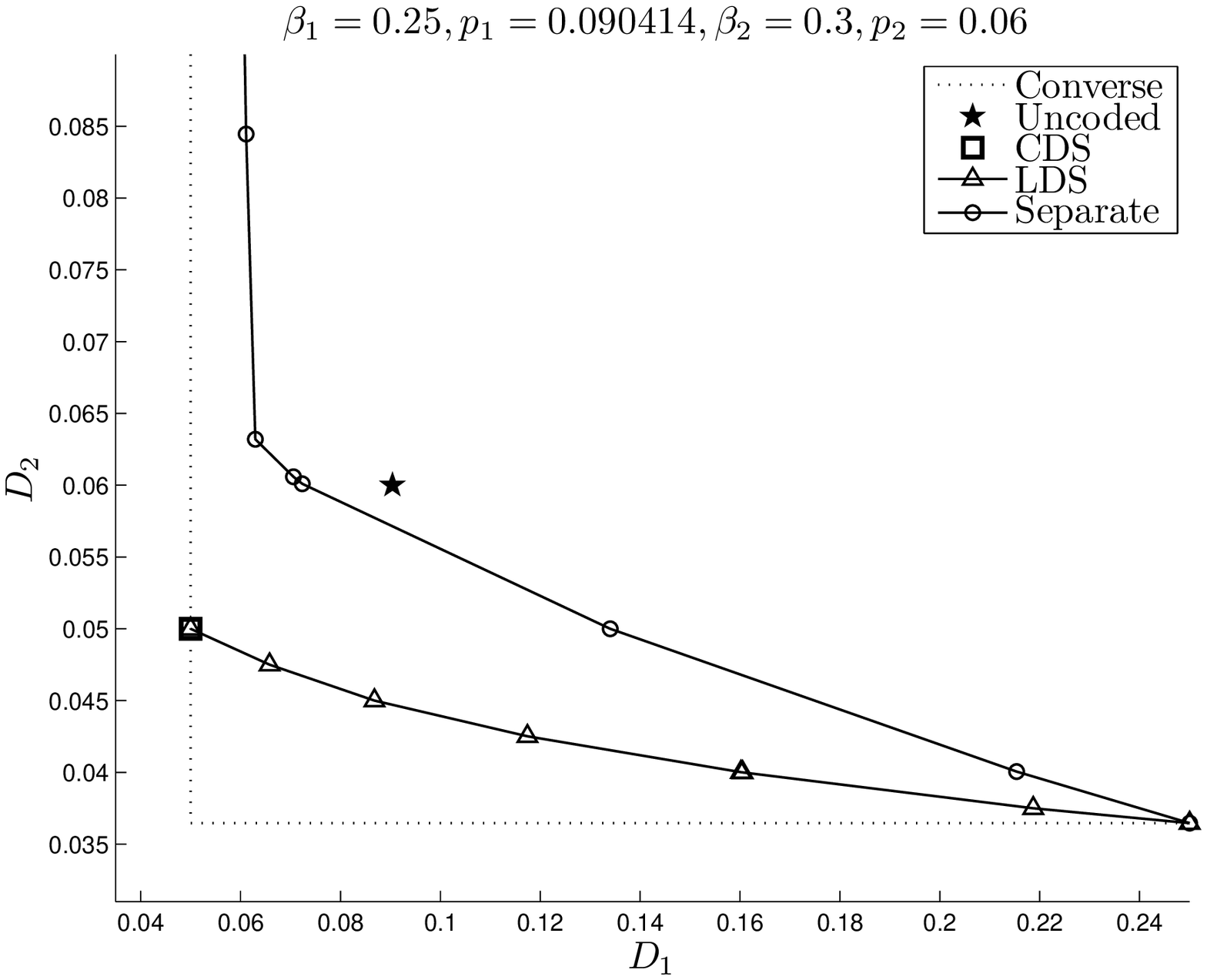}
}
\subfigure[]{
\includegraphics[width=0.46\columnwidth]{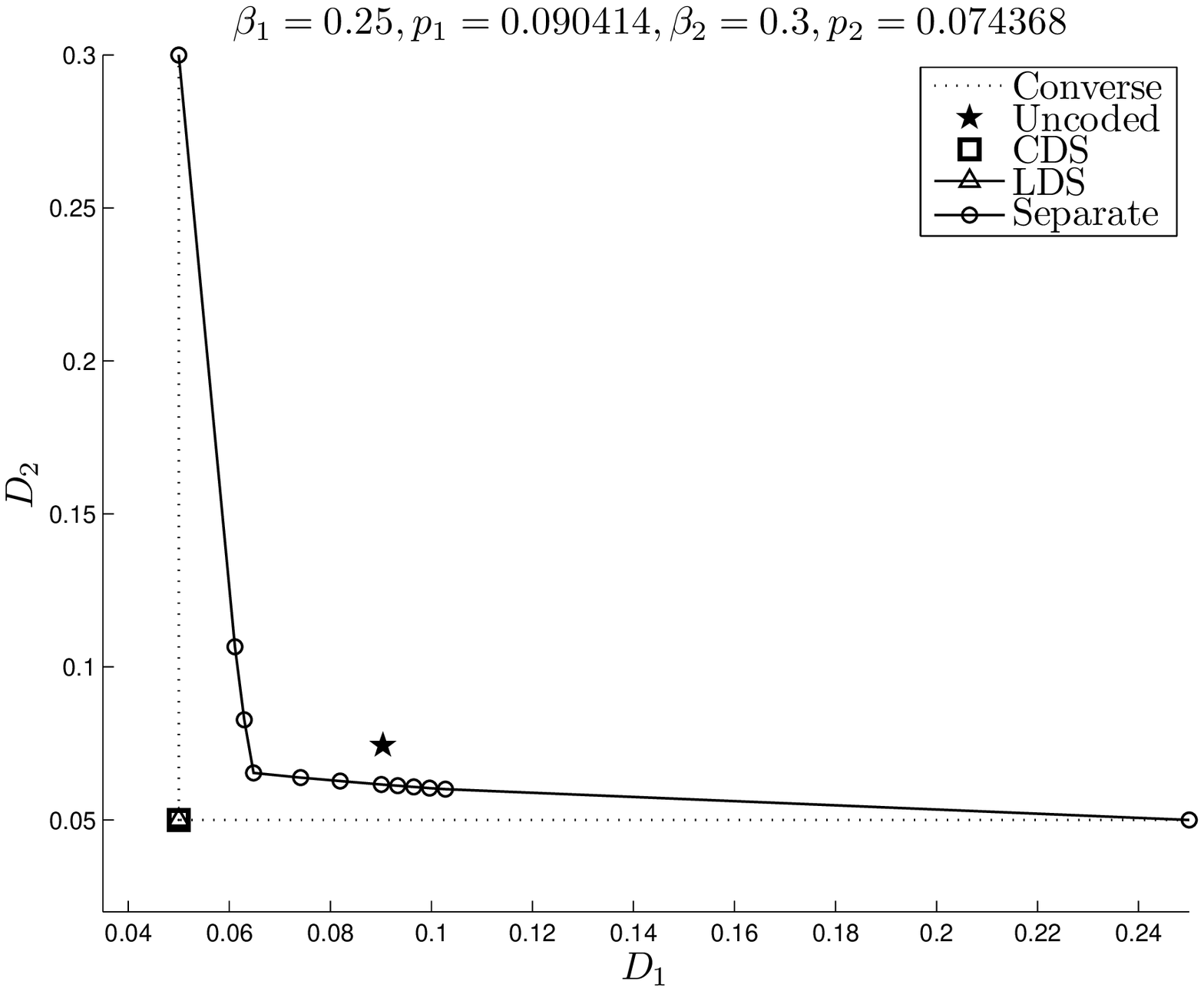}
}
\subfigure[]{
\includegraphics[width=0.46\columnwidth]{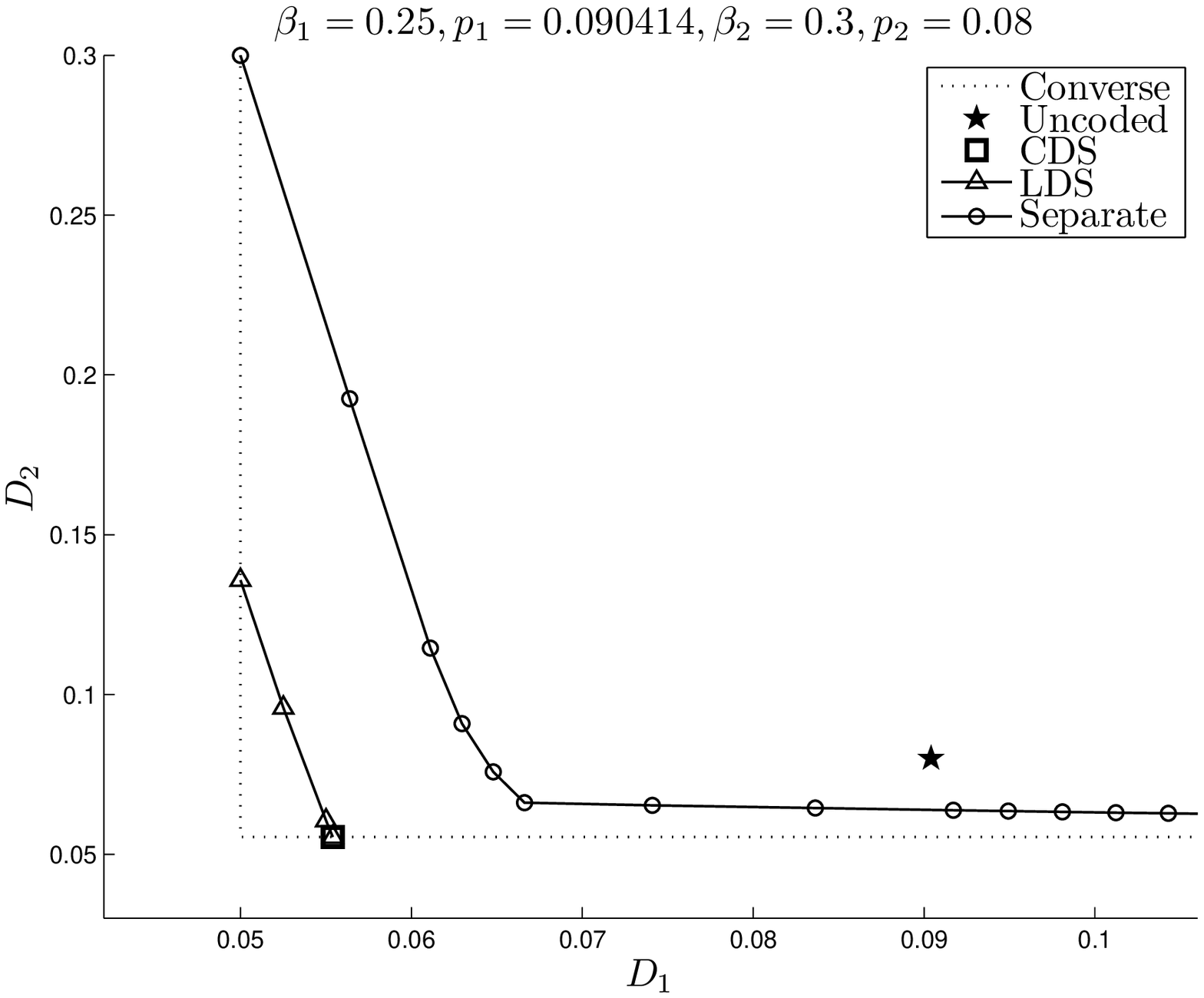}
}
\subfigure[]{
\includegraphics[width=0.46\columnwidth]{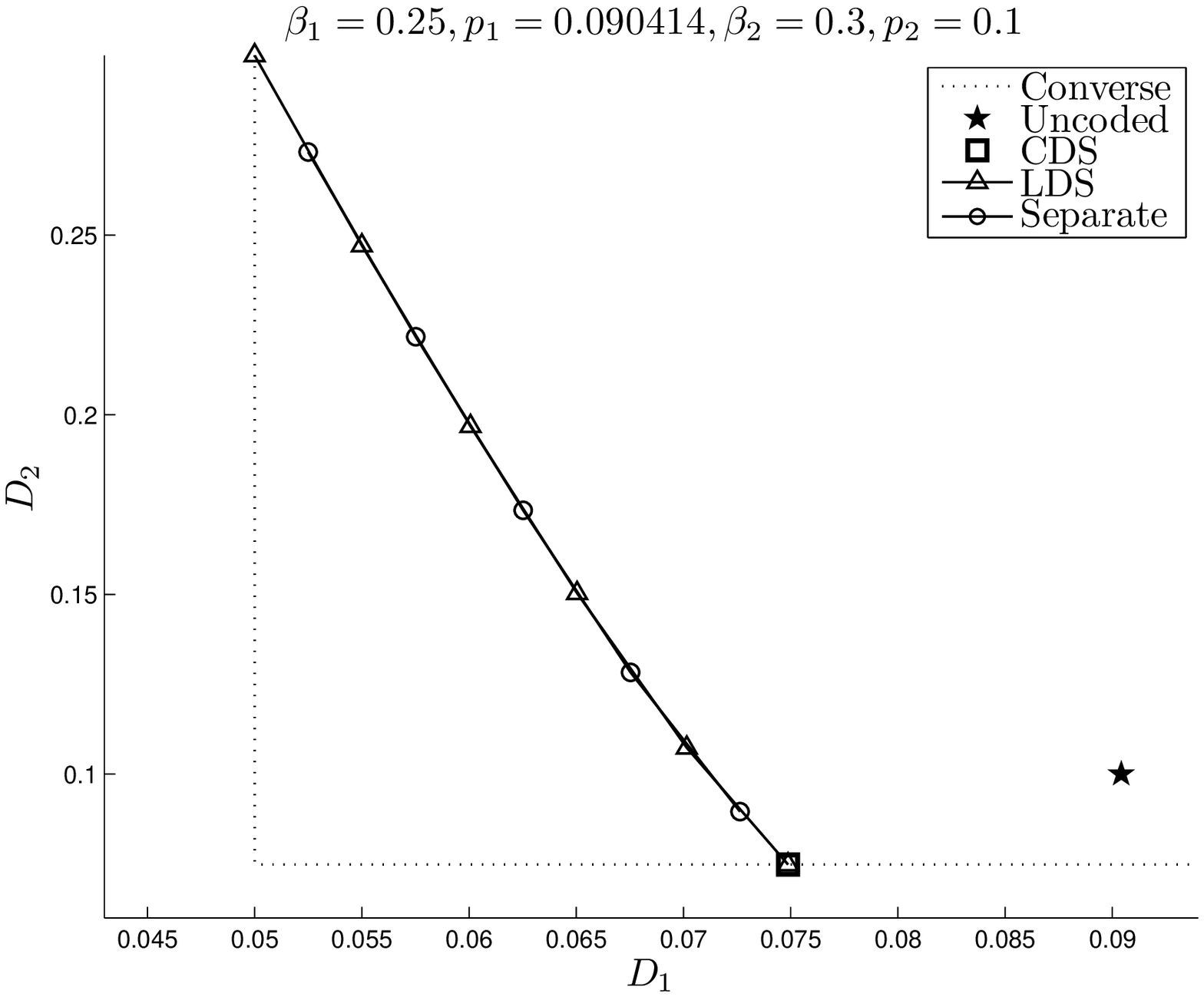}
}
\subfigure[]{
\includegraphics[width=0.46\columnwidth]{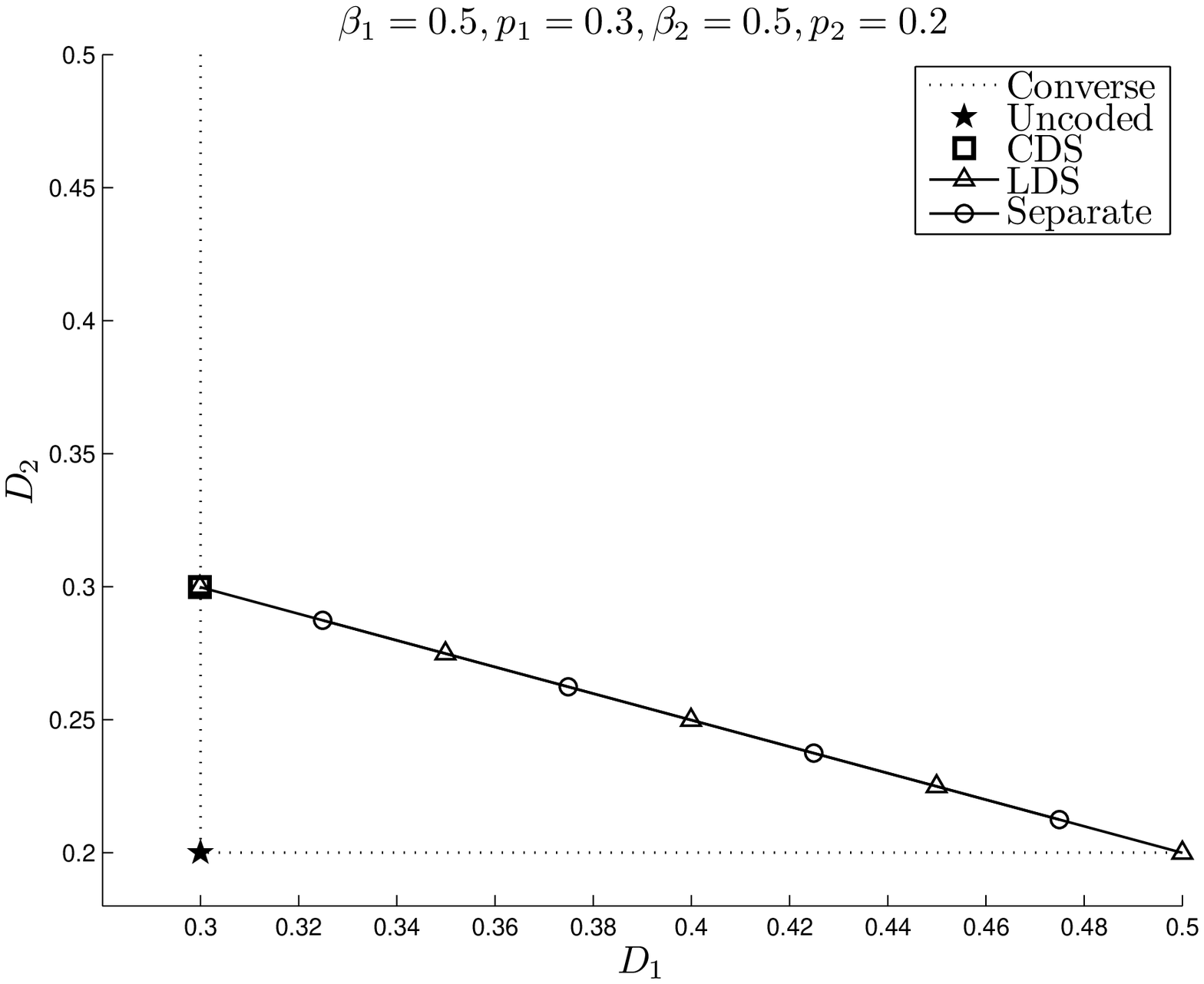}
}
\subfigure[]{
\includegraphics[width=0.46\columnwidth]{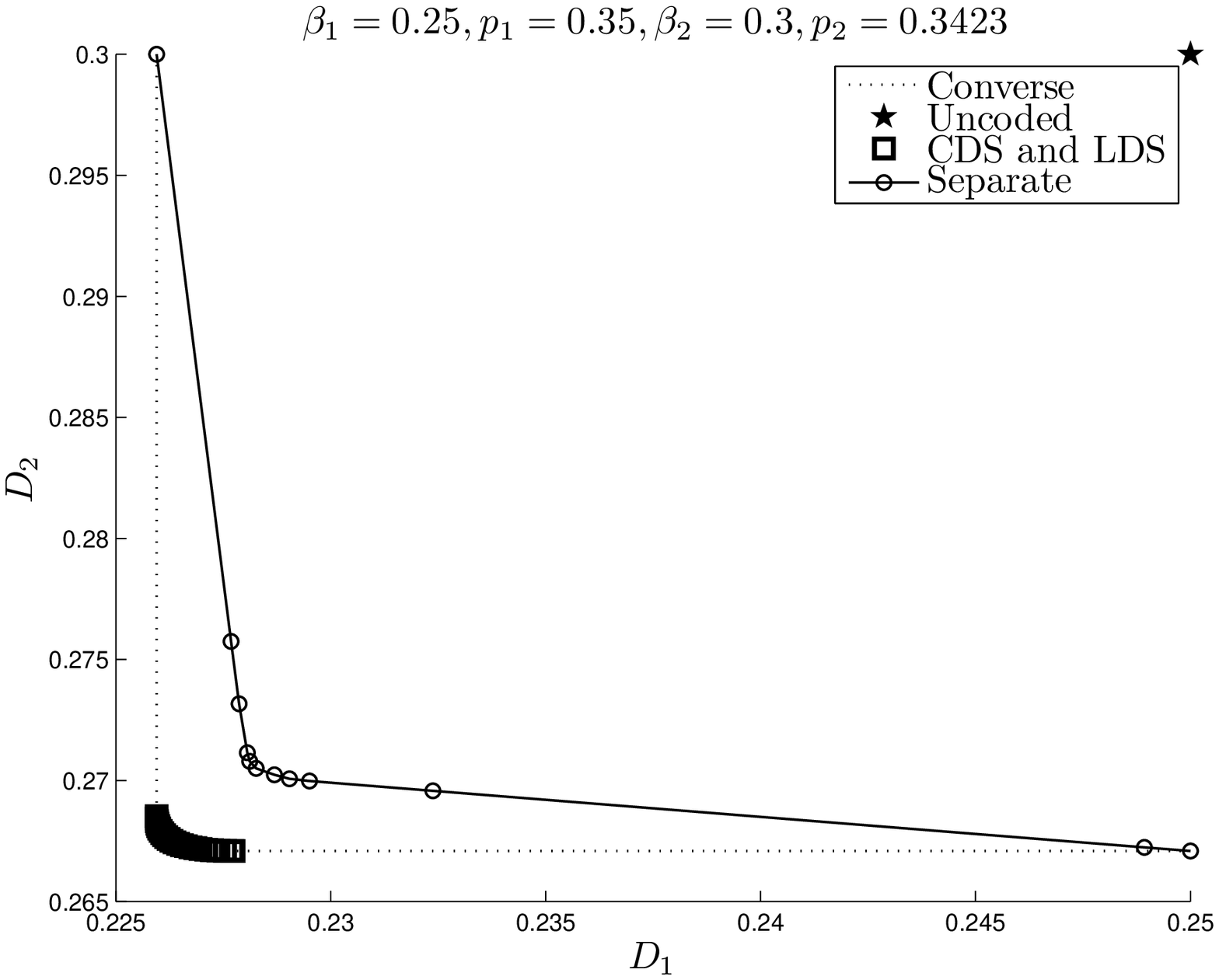}
}
\caption{Performance comparison for binary sources and channels. In (a)-(d), $\beta_1$, $\beta_2$, and $p_1$ are fixed, and as $p_2$ increases, how all the schemes compare changes. In (e), uncoded transmission is optimal. In (f), CDS and consequently LDS is the best. It is also noteworthy that it touches both trivial converse bounds simultaneously.}\label{figr:Bin}
\end{figure}

%%%%%%%%%%%%%%%%%%%%%%%%%%%%%%%%%%%%%%%%%%%%%%%%%%%%
\subsection{Performance Comparisons for the Bandwidth Matched Case: $\kappa=1$.}
%%%%%%%%%%%%%%%%%%%%%%%%%%%%%%%%%%%%%%%%%%%%%%%%%%%%

Analytical performance comparisons prove more difficult for the binary Hamming problem.
Even the question of which receiver should be designated as $c$ and which as $r$ is not straightforward to answer.
That is because (i) there is no power allocation parameter we can control, and (ii) even CDS can produce a curve which could achieve both $D_c=D_c^{WZ}(\kappa C_c)$ and $D_r=D_r^{WZ}(\kappa C_r)$, rather than a single best point.

It is also not clear that our choice of source random variables are the best.
As mentioned earlier, our main motivation in adopting the same test channel as in point-to-point coding for  LDS is its simplicity.
The alphabet size bounds in~\cite{SteinbergMerhav,TianDiggavi}, however, are much higher and therefore it might be possible to further improve the performance of LDS.

Using the same auxiliary random variables in separate coding gives us the following achievable result. We do not have a complete characterization of the distortion tradeoff.

\begin{lemma}
\label{lmma:BinarySeparate}
A distortion pair $(D_b, D_g)$ is achievable if there exist variables $0\leq q_b, q_g \leq 1$ and $0\leq \alpha_b,\alpha_g\leq \frac{1}{2}$ that satisfy
\begin{align}
q_b r(\alpha_b,\beta_b) &\leq \kappa [1-H_2(\theta\star p_b)] \;, \label{eqtn:BinarySeparateWorseChannel}\\
q_b r(\alpha_b,\beta_b) + [q_g r(\alpha_g,\beta_g) - q_b r(\alpha_b,\beta_g)]^+ &\leq  \kappa [H_2(\theta\star p_g) -H_2(p_g)] & {\rm if} X-Y_g-Y_b\;, \label{eqtn:BinarySeparateBetterSideBetterChannel}\\
q_g r(\alpha_g,\beta_g) + [ q_b r(\alpha_b,\beta_b) - q_g r(\alpha_g,\beta_b) ]^+ &\leq  \kappa [H_2(\theta\star p_g) -H_2(p_g)] & {\rm if} X-Y_b-Y_g\;. \label{eqtn:BinarySeparateBetterSideWorseChannel} \\
D_i &\leq q_i\min{\alpha_i,\beta_i} + (1-q_i)\beta_i, & i\in\{b,g\}
\end{align}
\end{lemma}
The proof is presented in Appendix~\ref{subs:App_lmma:BinarySeparate}.

The performance of the various schemes for certain source-channel pairs at rate $\kappa=1$ is presented in Figure~\ref{figr:Bin}. For LDS, the convex hull of two curves is shown, where in one $c=2, r=1$ and in the other $c=1, r=2$.
In Figures~\ref{figr:Bin}(a)-(d), the parameters $\beta_1$, $\beta_2$, and $p_1$ are fixed so that \eqref{eqtn:BinaryScheme0optimal} is satisfied for $k=1$, and $p_2$ is varying.
As $p_2$ increases, the collective behavior of the schemes dramatically changes. In Figure~\ref{figr:Bin}(a), $c=1,r=2$ is consistently the best choice among all schemes.
As the quality of the second channel decreases, and reaches the point where \eqref{eqtn:BinaryScheme0optimal} is also satisfied for $k=2$, CDS becomes optimal, as shown in Figure~\ref{figr:Bin}(b).
When $p_2$ is increased even further, as in Figure~\ref{figr:Bin}(c), $c=2,r=1$ becomes the better choice.
When $p_2$ reaches the point where the first receiver has access to both the better channel and the better side information, as in Figures~\ref{figr:Bin}(d) and (e), separate coding and LDS become identical as in the quadratic Gaussian case. However, uncoded transmission can still outperform the LDS as shown in Figure~\ref{figr:Bin}(e) for the case of trivial side information.
Finally, Figure~\ref{figr:Bin}(f) exemplifies the interesting phenomenon mentioned above, where CDS (and LDS) produces a curve, rather than a point, which happens to be the best.

%%%%%%%%%%%%%%%%%%%%%%%%%%%%%%%%%%%%%
%%%%%%%%%%%%%%%%%%%%%%%%%%%%%%%%%%%%%
\section{Conclusions and Future Work}
\label{sctn:Conclusions}
%%%%%%%%%%%%%%%%%%%%%%%%%%%%%%%%%%%%%
%%%%%%%%%%%%%%%%%%%%%%%%%%%%%%%%%%%%%

We proposed a layered coding scheme for the WZBC problem, and analyzed its distortion performance for the quadratic Gaussian and binary Hamming cases. Even though our scheme allows for arbitrary rate $\kappa$ channel uses per source symbol, the achievability regions are easiest to compute for $\kappa=1$. In fact, for the quadratic Gaussian case, we were able to derive closed form expressions for the entire distortion tradeoff and show that our layered scheme is always at least as good as (in fact, except for one certain case, always better than) separate coding. By numerical comparisons, we observed the same phenomenon for the binary Hamming case under the regime where all the test channels are constrained to be of the form which achieves the Wyner-Ziv rate-distortion function. On the other hand, our scheme may not always improve over the performance of uncoded transmission. This is not surprising, since when there is no (or trivial) side information, it is known that uncoded transmission is optimal. 

In an upcoming paper, we combine the digital scheme we proposed with uncoded transmission to extract the benefits of both methods.
In fact, as we show in a preliminary version~\cite{Deniz}, the hybrid scheme is more than the sum of its parts and distortions outside the convexification of the digital and analog regions are achievable.

%%%%%%%%%%%%%%%%%%%%%%%%%%%%%%%%%%%%%%%%%%%%%%%%%%%%%%%%%
%%%%%%%%%%%%%%%%%%%%%%%%%%%%%%%%%%%%%%%%%%%%%%%%%%%%%%%%%
%%%%%%%%%%%%%%%%%%%%%%%%%%%%%%%%%%%%%%%%%%%%%%%%%%%%%%%%%
\appendix
%%%%%%%%%%%%%%%%%%%%%%%%%%%%%%%%%%%%%%%%%%%%%%%%%%%%%%%%%
%%%%%%%%%%%%%%%%%%%%%%%%%%%%%%%%%%%%%%%%%%%%%%%%%%%%%%%%%
%%%%%%%%%%%%%%%%%%%%%%%%%%%%%%%%%%%%%%%%%%%%%%%%%%%%%%%%%

%%%%%%%%%%%%%%%%%%%%%%%%%%%%%%%%%%%%%%%%%%%%%%%%%%%%%%%
\subsection{Other Layered WZBC Schemes}
\label{subs:App:OtherSchemes}
%%%%%%%%%%%%%

The LDS that we focus on in this paper is only one of many possible layered coding schemes based on CDS and CDS with DPC. We shall briefly discuss these schemes. In all schemes, the source coding rates are the same as in LDS and only the channel coding rates differ.
\begin{itemize}
\item \emph{Scheme 1%
%CR-CR
}: This scheme is the simplest extension of CDS. The CL is encoded as in CDS. The RL is encoded on top of the CL. At both decoders, the RL is a source of interference while decoding the CL. Once the CL is decoded at the refinement receiver, its effect can be cancelled while decoding the RL. The acheivable channel rates are given by the next theorem.%
\begin{theorem}
Let ${\cal R}^c_{\rm 1}(\kappa)$ be the union of all $(R^c_{cc},R^c_{cr},R^c_{rr})$ for which there exist $U_c$ in some auxiliary alphabet $\mathcal{U}_c$ and $U\in\mathcal{U}$ with $U_c-U-(V_c,V_r)$ such that
\begin{align}
\label{eqtn:CRCR1}
R^c_{cc} &\leq \kappa I(U_c;V_c) \\
\label{eqtn:CRCR2}
R^c_{cr} &\leq \kappa I(U_c;V_r) \\
\label{eqtn:CRCR3}
R^c_{rr} &\leq \kappa I(U;V_r|U_c)\;.
\end{align}
Then ${\cal R}^c_{\rm 1} \subseteq {\cal C}_{\rm WZBC}(\kappa)$.
\end{theorem}
%\begin{remark}
%To see that CDS is indeed a special case of Scheme 1, it suffices to use the trivial superposition $U=U_c$.
%\end{remark}
\begin{proof}
Given random variables $U$ and $U_c$ such that $U_c-U-(V_c,V_r)$ and \eqref{eqtn:CRCR1}-\eqref{eqtn:CRCR3} are satisfied, each $U_c^m(i)$ in the CL channel codebook is chosen uniformly and independently from ${\cal T}_{\delta}^m(U_c)$.
Similarly, for each $i$, codewords $U^m(j'|i)$ to be transmitted over the channel are chosen uniformly and independently from ${\cal T}_{\delta'}^m(U|U_c)$.
It then follows from Corollary~\ref{corr:S0} that \eqref{eqtn:CRCR1} and \eqref{eqtn:CRCR2} are sufficient for successful decoding of both $Z_c^n(i)$ and $U_c^m(i)$ simultaneously at both decoders.
It also follows from standard arguments that \eqref{eqtn:CRCR3} is sufficient for reliable transmission of additional information with rate $R^c_{rr}$ to the refinement receiver.
\end{proof}

\item \emph{Scheme 2%
%CR-RC
}: The CL is encoded as in Scheme 1. The RL, however, is sent using dirty paper coding with the CL codeword as encoder CSI, and is decoded first.
\begin{theorem}
Let ${\cal R}^c_{\rm 2}(\kappa)$ be the union of all $(R^c_{cc},R^c_{cr},R^c_{rr})$ for which there exist $U_c\in{\cal U}_c$, $U_r\in{\cal U}_r$, and $T\in{\cal T}$ with $T-(U_r,U_c)-(V_r, V_c)$ and $(U_r,U_c)-U-(V_r, V_c)$ such that
\begin{align}
\label{eqtn:CRRC1}
R^c_{cc} &\leq \kappa I(U_c;V_c) \\
\label{eqtn:CRRC2}
R^c_{cr} &\leq \kappa I(U_c;T,V_r) \\
\label{eqtn:CRRC3}
R^c_{rr} &\leq \kappa [I(T;V_r) - I(T;U_c)]\;.
\end{align}
Then ${\cal R}^c_{\rm 2}(\kappa)\subseteq{\cal C}_{\rm WZBC}(\kappa)$.
\end{theorem}
%\begin{remark}
%To specialize Scheme 2 to CDS, one needs to set $U=U_c$ and trivially pick $T$ as a random variable independent of $U_c$ and $U_r$.
%\end{remark}
\begin{proof}
Since RL is to be sent by separate source and channel codes, the channel coding part can proceed as in standard dirty-paper coding (cf.~\cite{Gelfand}), if \eqref{eqtn:CRRC3} is satisfied.
Note that as in Corollary~\ref{corr:Scheme0_DPC}, the auxiliary codeword $T^m$ can also be decoded in the process of decoding the RL.
With high probability, this codeword is typical with the CL codeword $U_c^m$ in addition to $V_r^m$. Subsequently, for decoding the CL, the channel output at the $r$ decoder can be taken to be a pair $(V_r^m, T^m)$. Therefore, as in Scheme 1,
$Z_c^n(i)$ can be successfully decoded given that \eqref{eqtn:CRRC1}  and \eqref{eqtn:CRRC2} hold.
\end{proof}

\item \emph{Scheme 3%
%RC-RC
}: The encoding is performed as in LDS, but the decoding order is reversed. Since RL is decoded first at the $r$ receiver, the CL codeword purely acts as noise. But the $r$ decoder then has access to the RL codeword. So for that receiver, the CSI is also available at the decoder.
The following theorem makes use of these observations.

\begin{theorem}
Let ${\cal R}^c_{\rm 3}(\kappa)$ be the union of all $(R^c_{cc},R^c_{cr},R^c_{rr})$ for which there exist $U_c\in{\cal U}_c$, $U_r\in{\cal U}_r$, and $T\in{\cal T}$ with $T-(U_r,U_c)-(V_r, V_c)$ and $(U_r,U_c)-U-(V_r, V_c)$ such that
\begin{align}
\label{eqtn:RCRC1}
R^c_{cc} &\leq \kappa [I(T;V_c) - I(T;U_r)] \\
\label{eqtn:RCRC2}
R^c_{cr} &\leq \kappa I(T;V_r|U_r)\\
\label{eqtn:RCRC3}
R^c_{rr} &\leq \kappa I(U_r;V_r) \;.
\end{align}
Then ${\cal R}^c_{\rm 3}(\kappa)\subseteq{\cal C}_{\rm WZBC}(\kappa)$.
\end{theorem}
%\begin{remark}
%Similar to the LDS, choosing a trivial $U_r$ and setting $T=U$ reduces this scheme to CDS.
%\end{remark}
\begin{proof}
Since RL is both encoded and decoded first, \eqref{eqtn:RCRC3} is necessary and sufficient for successful decoding of $U_r^m$.
Once $U_r^m$ is decoded, the channel between CL and receiver $r$ reduces to one with input $U_c^m$, output $(V_r^m,U_r^m)$, and CSI $U_r^m$.
It then follows from Theorem~\ref{thrm:S0_DPC} that \eqref{eqtn:RCRC1} and \eqref{eqtn:RCRC2} suffices for reliable transmission of $Z_c^m$.
Note that the right-hand side of \eqref{eqtn:RCRC2} is equivalent to $I(T;U_r,V_r) - I(T;U_r)$.
\end{proof}
\end{itemize}

We now present partial analytical results comparing performances of all the layered schemes.
 
\begin{lemma}
It is always true that ${\cal R}^c_2(\kappa) \subseteq {\cal R}^c_1(\kappa)$. 
Thus Scheme 1 is superior to Scheme 2.
\end{lemma}
\begin{proof}
It suffices to prove the lemma for $\kappa=1$.
Let $(R_{cc}^c,R_{cr}^c,R_{rr}^c)\in{\cal R}^c_2(1)$. Then there must exist $U_c^{(1)}$, $U_r^{(1)}$, $T$, and $U$ with $T-(U_c^{(1)},U_r^{(1)})-(V_c,V_r)$ and $(U_c^{(1)},U_r^{(1)})-U-(V_c,V_r)$ so that \eqref{eqtn:CRRC1}-\eqref{eqtn:CRRC3} are satisfied.
Now define $U_c^{(2)}=U$ and let
\begin{eqnarray*}
R_{cc}^{(1)} & = & I(U_c^{(1)};V_c) \\
R_{cr}^{(1)} & = & I(U_c^{(1)};V_r) \\
R_{rr}^{(1)} & = & I(U;V_r|U_c^{(1)}) 
\end{eqnarray*}
and
\begin{eqnarray*}
R_{cc}^{(2)} & = & I(U_c^{(2)};V_c) \;\;=\;\; I(U;V_c) \\
R_{cr}^{(2)} & = & I(U_c^{(2)};V_r) \;\;=\;\; I(U;V_r) \\
R_{rr}^{(2)} & = & I(U;V_r|U_c^{(2)}) \;\;=\;\; 0 \; .
\end{eqnarray*}
By definition, both $(R_{cc}^{(1)},R_{cr}^{(1)},R_{rr}^{(1)})$ and $(R_{cc}^{(2)},R_{cr}^{(2)},R_{rr}^{(2)})$ belong to ${\cal R}^c_1(1)$.
So does any convex combination of the two triplets. That is because if we define $Q\sim\Ber(\lambda)$, so that 
\[
p(q,u_c,u,v_c,v_r) = p(u,v_c,v_r)p(q)p(u_c|u,q)
\]
we can then write any convex combination as 
\begin{eqnarray*}
R_{cc}^{(\lambda)} & = & I(U_c^{(Q)};V_c|Q) \;\;=\;\; I(U_c^{(Q)},Q;V_c) \\
R_{cr}^{(\lambda)} & = & I(U_c^{(Q)};V_r|Q) \;\;=\;\; I(U_c^{(Q)},Q;V_r) \\
R_{rr}^{(\lambda)} & = & I(U;V_r|U_c^{(Q)},Q) \; .
\end{eqnarray*}
Defining $U_c^{(\lambda)}=(U_c^{(Q)},Q)$, one can see that 
$(R_{cc}^{(\lambda)},R_{cr}^{(\lambda)},R_{rr}^{(\lambda)}) \in {\cal R}^c_1(1)$.

It is clear that 
\begin{equation}
\label{eqtn:CRCRbetterCRRC1}
R_{cr}^{(1)}\leq I(U_c^{(1)};T,V_r) \; .
\end{equation}
It also follows from the Markov chain $(U_c^{(1)},T)-U-V_r$ that
\begin{equation}
\label{eqtn:CRCRbetterCRRC7}
I(U_c^{(1)},T;V_r) \leq I(U;V_r)\; .
\end{equation}
A fact which is not as obvious is
\begin{equation}
\label{eqtn:CRCRbetterCRRC3}
R_{cr}^{(1)} + R_{rr}^{(1)} \geq I(U_c^{(1)};T,V_r) \; .
\end{equation}
Towards proving \eqref{eqtn:CRCRbetterCRRC3}, we observe using \eqref{eqtn:CRCRbetterCRRC7} that 
\begin{eqnarray}
I(U;V_r) & \geq & I(U_c^{(1)},T;V_r) \nonumber \\
& = & I(U_c^{(1)};V_r|T) + I(T;V_r) \nonumber \\
\label{eqtn::CRCRbetterCRRC8}
& = & I(U_c^{(1)};T,V_r) + I(T;V_r) - I(T;U_c^{(1)}) \; .
\end{eqnarray}
But since $R_{cr}^{(1)}+R_{rr}^{(1)}= I(U;V_r)$, this yields \eqref{eqtn:CRCRbetterCRRC3} directly.

Next, we choose $\lambda$ so that
\[
R_{cr}^{(\lambda)} = I(U_c^{(1)};T,V_r) \; .
\]
That this can always be done follows from \eqref{eqtn:CRCRbetterCRRC1} and \eqref{eqtn:CRCRbetterCRRC3} together with the observation that $R_{cr}^{(1)} + R_{rr}^{(1)} = R_{cr}^{(2)}$.
We then simultaneously have 
\begin{eqnarray}
\label{eqtn:CRCRbetterCRRC4}
R_{cc}^{(\lambda)} & \geq & R_{cc}^c \\
\label{eqtn:CRCRbetterCRRC5}
R_{cr}^{(\lambda)} & \geq & R_{cr}^c \\
\label{eqtn:CRCRbetterCRRC6}
R_{rr}^{(\lambda)} & \geq & R_{rr}^c \; .
\end{eqnarray}
Here, \eqref{eqtn:CRCRbetterCRRC4} follows from the fact that $R_{cc}^c\leq I(U_c^{(1)};V_c) =  R_{cc}^{(1)}\leq R_{cc}^{(2)}$.
The fact that $R_{cr}^c\leq I(U_c^{(1)};T,V_r)=R_{cr}^{(\lambda)}$ yields \eqref{eqtn:CRCRbetterCRRC5}.
Finally, \eqref{eqtn:CRCRbetterCRRC6} follows because
\begin{eqnarray}
R_{rr}^{(\lambda)} & = & I(U;V_r) - R_{cr}^{(\lambda)} \nonumber \\
& = & I(U;V_r) - I(U_c^{(1)};T,V_r) \nonumber \\
\label{eqtn:CRCRbetterCRRC9}
& \geq & I(T;V_r) - I(T;U_c^{(1)}) \\
& \geq & R_{rr}^c \nonumber
\end{eqnarray}
where we used \eqref{eqtn::CRCRbetterCRRC8} in showing \eqref{eqtn:CRCRbetterCRRC9}.
\end{proof}

It is also easy to show that under the regime where $U=U_c+U_r$ where $+$ is an appropriately defined addition operation with an inverse, i.e., $U_r=U-U_c$, and $U_c$ and $U_r$ are independent, Scheme 1 becomes a special case of LDS. 
Thus, for both the quadratic Gaussian and the binary Hamming cases, LDS performs at least as well as Scheme 1.
To prove this claim, it suffices to pick $T=U_c$ in LDS, which achieves the performance
\begin{eqnarray*}
R_{cc}^c & = & \kappa[I(T;V_c) - I(T;U_r)]\\
& = & \kappa I(U_c;V_c) \\
R_{cr}^c & = & \kappa[I(T;V_r) - I(T;U_r)]\\
& = & \kappa I(U_c;V_r) \\
R_{rr}^c & = & \kappa I(U_r;T,V_r)\\
 & = & \kappa[I(U_r;U_c) +I(U_r;V_r|U_c)]\\
 & = & \kappa I(U_r+U_c;V_r|U_c)\\
 & = & \kappa I(U;V_r|U_c) 
\end{eqnarray*}
making Scheme 1 is a special case of LDS.

We can also compare the performances of Scheme 3 and LDS for the quadratic Gaussian case with $\kappa=1$.
Using the same random variables as in LDS, \eqref{eqtn:RCRC1}-\eqref{eqtn:RCRC3} translate to the achievability of
\begin{eqnarray}
\label{eqtn:GaussianRCRC1}
R_{cc}^c & = & I(\gamma U_r + U_c;U_c+U_r+W_c) - I(\gamma U_r + U_c;U_r)  \nonumber \\
& = & \frac{1}{2}\log \frac{1+\frac{P}{\mathbf{W}_c}}{1 + \bar{\nu} P \left(\frac{\gamma^2}{\nu P} + \frac{(1-\gamma)^2}{\mathbf{W}_c}\right)} \\
\label{eqtn:GaussianRCRC2}
R_{cr}^c & = & I(\gamma U_r + U_c;U_c+U_r+W_r|U_r) \nonumber \\
& = & I(U_c;U_c+W_r) \nonumber \\
& = & \frac{1}{2}\log \left( 1 + \frac{\nu P}{\mathbf{W}_r}\right) \\
\label{eqtn:GaussianRCRC3}
R_{rr}^c & = & I(U_r;U_c+U_r+W_r) \nonumber \\
& = & \frac{1}{2}\log \left( 1 + \frac{\bar{\nu} P}{\nu P+\mathbf{W}_r}\right) 
\end{eqnarray}
where \eqref{eqtn:GaussianRCRC1} follows from \eqref{eqtn:GaussianRCCR1}.
Since the choice of $\gamma$ affects only $R_{cc}^c$, it can be picked so as to maximize $R_{cc}^c$.
In fact, this choice coincides with Costa's optimal $\gamma$ for the point-to-point channel between $U_c$ and $V_c$, where the CSI $U_r$ is available at the encoder~\cite{Costa}. 
In other words, the optimal choice is given by (cf.~\cite[Equation~(7)]{Costa})
\[
\gamma = \frac{\nu P}{\nu P+\mathbf{W}_c}
\]
yielding
\begin{equation}
\label{eqtn:GaussianRCRC4}
R_{cc}^c = \frac{1}{2} \log \left(1+\frac{\nu P}{\mathbf{W}_c}\right) \; .
\end{equation}
Also note that $R_{cr}^c+R_{rr}^c=\frac{1}{2}\log\left(1+\frac{P}{\mathbf{W}_r}\right)$, thereby keeping \eqref{eqtn:GaussianDc2} and \eqref{eqtn:GaussianDr2} valid.
That is, 
\begin{eqnarray}
\label{eqtn:GaussianRCRC_Dc}
D_c & = &  \frac{\mathbf{N}_c\mathbf{W}_c}{\nu P + \mathbf{W}_c} \\
\label{eqtn:GaussianRCRC_Dr}
D_r & = & \frac{\mathbf{N}_r}{1+\mathbf{N}_r\left[\frac{1}{D_c}-\frac{1}{\mathbf{N}_c}\right]} \cdot \frac{\nu P + \mathbf{W}_r}{P + \mathbf{W}_r} \; .
\end{eqnarray}
Solving for $\nu$ in \eqref{eqtn:GaussianRCRC_Dc} and substituting it in \eqref{eqtn:GaussianRCRC_Dr} yields
\begin{equation}
\label{eqtn:GaussianRCRC_DrFinal}
D_r = \frac{\mathbf{N}_r\mathbf{W}_r}{P+\mathbf{W}_r} \cdot \frac{D_c\mathbf{N}_c + \frac{\mathbf{N}_c\mathbf{W}_c}{\mathbf{W}_r}(\mathbf{N}_c-D_c)}{D_c\mathbf{N}_c + \mathbf{N}_r(\mathbf{N}_c-D_c)} \; .
\end{equation}
for the entire range
\[
\frac{\mathbf{N}_c\mathbf{W}_c}{P+\mathbf{W}_c} \leq D_c \leq \mathbf{N}_c \; .
\]

\begin{lemma}
\label{lmma:RCCRbetterRCRC}
For the quadratic Gaussian problem with $\kappa=1$, the performance of LDS is superior to that of Scheme 3.
\end{lemma}
\begin{proof}

Let us first compare \eqref{eqtn:GaussianRCRC_DrFinal}  to \eqref{eqtn:GaussianRCCR_Dr} for the $\mathbf{W}_c\geq \mathbf{W}_r$ case.
We shall show for all $\frac{\mathbf{N}_c\mathbf{W}_c}{P+\mathbf{W}_c}\leq D_c\leq\mathbf{N}_c$ that
\[
\frac{\mathbf{N}_r\mathbf{N}_c^2}{D_c\mathbf{N}_c 
+\mathbf{N}_r (\mathbf{N}_c-D_c)} \cdot \frac{\mathbf{W}_r D_c}{(\mathbf{W}_r-\mathbf{W}_c)\mathbf{N}_c+(P + \mathbf{W}_c)D_c} 
\leq \frac{\mathbf{N}_r\mathbf{W}_r}{P+\mathbf{W}_r} \cdot \frac{D_c\mathbf{N}_c + \frac{\mathbf{N}_c\mathbf{W}_c}{\mathbf{W}_r}(\mathbf{N}_c-D_c)}{D_c\mathbf{N}_c + \mathbf{N}_r(\mathbf{N}_c-D_c)}
\]
or equivalently that
\begin{equation}
\label{eqtn:CompareRCs1}
D_c\mathbf{N}_c(P+\mathbf{W}_r) \leq \left(D_c + \frac{\mathbf{W}_c}{\mathbf{W}_r}(\mathbf{N}_c-D_c)\right)\Big((\mathbf{W}_r-\mathbf{W}_c)\mathbf{N}_c+(P + \mathbf{W}_c)D_c\Big) \; .
\end{equation}
Adding $D_c\mathbf{N}_c(\mathbf{W}_c-\mathbf{W}_r)$ to both sides of \eqref{eqtn:CompareRCs1} yields
\begin{equation}
\label{eqtn:CompareRCs2}
D_c\mathbf{N}_c(P+\mathbf{W}_c) \leq \left(D_c + \frac{\mathbf{W}_c}{\mathbf{W}_r}(\mathbf{N}_c-D_c)\right)(P + \mathbf{W}_c)D_c + \frac{\mathbf{W}_c}{\mathbf{W}_r}(\mathbf{N}_c-D_c)(\mathbf{W}_r-\mathbf{W}_c)\mathbf{N}_c
\; .
\end{equation}
Taking the first term on the right-hand side of \eqref{eqtn:CompareRCs2} to the left-hand side, we obtain
\[
D_c(P+\mathbf{W}_c)(\mathbf{N}_c-D_c) \left(1-\frac{\mathbf{W}_c}{\mathbf{W}_r}\right) \leq \frac{\mathbf{W}_c}{\mathbf{W}_r}(\mathbf{N}_c-D_c)(\mathbf{W}_r-\mathbf{W}_c)\mathbf{N}_c
\]
or equivalently
\[
D_c(P+\mathbf{W}_c) \geq \mathbf{W}_c\mathbf{N}_c 
\]
which is guaranteed. Equality is satisfied in only three trivial cases: (i) When $D_c=D_c^{WZ}(C_c)$, which coincides with CDS, (ii) when $\mathbf{W}_c=\mathbf{W}_r$, and (iii) when $D_c=\mathbf{N}_c$, which should be excluded if $D_c^{\max}<\mathbf{N}_c$.

As for the $\mathbf{W}_c< \mathbf{W}_r$ case, to prove that LDS is superior, we need to show 
\[
\frac{\mathbf{N}_r\mathbf{N}_c^2}{D_c\mathbf{N}_c 
+\mathbf{N}_r (\mathbf{N}_c-D_c)} \cdot \frac{\mathbf{W}_c}{P+\mathbf{W}_c} 
\leq \frac{\mathbf{N}_r\mathbf{W}_r}{P+\mathbf{W}_r} \cdot \frac{D_c\mathbf{N}_c + \frac{\mathbf{N}_c\mathbf{W}_c}{\mathbf{W}_r}(\mathbf{N}_c-D_c)}{D_c\mathbf{N}_c + \mathbf{N}_r(\mathbf{N}_c-D_c)}
\]
or equivalently that
\begin{equation}
\label{eqtn:CompareRCs3}
\frac{\mathbf{N}_c\mathbf{W}_c}{P+\mathbf{W}_c} 
\leq \frac{D_c\mathbf{W}_r+ \mathbf{W}_c(\mathbf{N}_c-D_c)}{P+\mathbf{W}_r} \; .
\end{equation}
Rearranging \eqref{eqtn:CompareRCs3}, we have
\[
\mathbf{N}_c\mathbf{W}_c(P+\mathbf{W}_r)
\leq (P+\mathbf{W}_c)\Big(D_c(\mathbf{W}_r- \mathbf{W}_c)+\mathbf{W}_c\mathbf{N}_c\Big)
\]
which is once again equivalent to 
\[
D_c(P+\mathbf{W}_c) \geq \mathbf{W}_c\mathbf{N}_c \; .
\]
Equality in this case is satisfied if and only if $D_c=D_c^{WZ}(C_c)$.
\end{proof}

To summarize, for the quadratic Gaussian case with $\kappa=1$, LDS is provably the best. 
In the binary Hamming case, however, LDS is better than both Scheme 1 and Scheme 2, but an analytical comparison with Scheme 3 eluded us. 
Nevertheless, with an extensive set of numerical evaluations, we did not encounter a single case that Scheme 3 was better than LDS for the binary Hamming case with $\kappa=1$.

%%%%%%%%%%%%%%%%%%%%%%%%%%%%%%%%%%%%%%%%%%%%%%%%%%%%%
\subsection{Proof of Lemma~\ref{lmma:ClosedFormRCCR}}
\label{subs:App_lmma:ClosedFormRCCR}

It follows from \eqref{eqtn:GaussianDc2} and \eqref{eqtn:GaussianDr2} that by varying $\nu$ and $\gamma$, we obtain the tradeoff 
\begin{eqnarray}
\label{eqtn:GaussianDc_explicit}
D_c & = & \vari{N}_c \frac{Pa(\nu,\gamma)+ \vari{W}_c}{P + \vari{W}_c} \\
\label{eqtn:GaussianDr_explicit}
D_r & = & \frac{\vari{N}_r}{1+\vari{N}_r\left[\frac{1}{D_c}-\frac{1}{\vari{N}_c}\right]} \cdot \frac{1}{1+\frac{Pb(\nu,\gamma)}{\vari{W}_r}}
\end{eqnarray}
where
\begin{eqnarray*}
a(\nu,\gamma) &  =  & \bar{\nu}\left(\frac{\vari{W}_c}{\nu P}\gamma^2+(1-\gamma)^2\right) \\
b(\nu,\gamma) &  =  & \bar{\nu}\left(\frac{\vari{W}_r}{\nu P}\gamma^2+(1-\gamma)^2\right) \; .
\end{eqnarray*}
We next fix $D_c$, which, in turn, fixes $a(\nu,\gamma)$ as 
\begin{equation}
\label{eqtn:aDc}
a(\nu,\gamma) = \frac{D_c[P+\vari{W}_c]-\vari{W}_c\vari{N}_c}{\vari{N}_c P }
\end{equation}
and minimize $D_r$, which reduces to maximizing $b(\nu,\gamma)$.
Since neither $R^c_{cc}$ nor $R^c_{cr}$ can be negative, we need both $a(\nu,\gamma)\leq 1$ and $b(\nu,\gamma) \leq 1$ to be satisfied.
The former requirement is guaranteed because we naturally limit ourselves to $D_c\leq \vari{N}_c$.
The latter, on the other hand, becomes vacuous since rewriting 
\eqref{eqtn:GaussianPhiCondition2} gives 
\begin{equation}
\label{eqtn:RCCRConstraints3}
b(\nu,\gamma) \leq \frac{\vari{N}_c[Pa(\nu,\gamma)+\vari{W}_c]-\vari{N}_r\vari{W}_r[1-a(\nu,\gamma)]}{\vari{N}_c[Pa(\nu,\gamma)+\vari{W}_c]+P\vari{N}_r[1-a(\nu,\gamma)]} 
\end{equation}
whose right-hand side is always less than or equal to 1.

Now if $\vari{W}_c\geq \vari{W}_r$, we always have $a(\nu,\gamma)\geq b(\nu,\gamma)$ since
\[
b(\nu,\gamma) = a(\nu,\gamma) - \frac{\bar{\nu}\gamma^2}{\nu P} [\vari{W}_c-\vari{W}_r] \; .
\]
Thus, among all choices of $\gamma$ and $\nu$ which satisfy \eqref{eqtn:aDc}, the one that potentially minimizes $D_r$ is $\gamma=0$ and
\[
\nu = \left(1-\frac{D_c}{\vari{N}_c}\right)\left(1+\frac{\vari{W}_c}{P}\right) \; .
\]
That is because with this choice we have $b(\nu,\gamma)=a(\nu,\gamma)$. 
It then remains to check \eqref{eqtn:RCCRConstraints3}, which can be written after some algebra as 
\[
\vari{N}_c\vari{N}_r[\vari{W}_c-\vari{W}_r] \geq D_c [P+\vari{W}_c][\vari{N}_r-\vari{N}_c] \; .
\]
This is granted if $\vari{N}_r\leq\vari{N}_c$ and is equivalent to
\begin{equation}
\label{eqtn:RCCRConstraints4}
D_c \leq \frac{\vari{N}_c\vari{N}_r[\vari{W}_c-\vari{W}_r]}{[P+\vari{W}_c][\vari{N}_r-\vari{N}_c]} 
\end{equation}
if $\vari{N}_r>\vari{N}_c$.
The constraint \eqref{eqtn:RCCRConstraints4}, on the other hand, 
is in effect only if 
\[
\vari{N}_c(P+\vari{W}_c) < \vari{N}_r(P+\vari{W}_r) 
\]
for otherwise, it is trivially satisfied because $D_c\leq \vari{N}_c$.
Substituting $b(\nu,\gamma)=a(\nu,\gamma)$ in \eqref{eqtn:GaussianDr} yields
\[
D_r = \frac{\vari{N}_r\vari{W}_r \vari{N}_c^2D_c}{\Big(D_c\vari{N}_c +\vari{N}_r (\vari{N}_c-D_c)\Big)\Big((\vari{W}_r-\vari{W}_c)\vari{N}_c+(P + \vari{W}_c)D_c\Big)} \; .
\]

On the other hand, if $\vari{W}_c< \vari{W}_r$, it is more helpful to write
\[
b(\nu,\gamma) = \frac{\vari{W}_r}{\vari{W}_c} a(\nu,\gamma) - \bar{\nu} (1-\gamma)^2 \left[ \frac{\vari{W}_r}{\vari{W}_c}-1\right] 
\]
as this reveals $b(\nu,\gamma)\leq \frac{\vari{W}_r}{\vari{W}_c} a(\nu,\gamma)$.
Thus, the optimal choice of parameters is potentially $\gamma=1$ and
\[
\nu = \frac{\vari{W}_c\vari{N}_c}{D_c[P+\vari{W}_c]}
\]
provided this choice satisfies \eqref{eqtn:RCCRConstraints3}.
Once again, after some algebra, that translates to
\[
D_c \leq \frac{\vari{N}_cP[\vari{W}_c\vari{N}_c-\vari{W}_r\vari{N}_r]+\vari{W}_c\vari{W}_r\vari{N}_c[\vari{N}_c-\vari{N}_r]}{[P+\vari{W}_c][\vari{N}_c-\vari{N}_r]\vari{W}_r} 
\]
Substituting $b(\nu,\gamma)=\frac{\vari{W}_r}{\vari{W}_c}a(\nu,\gamma)$ in \eqref{eqtn:GaussianDr} yields
\[
D_r = \frac{\vari{N}_r\vari{N}_c^2\vari{W}_c}{\Big(D_c\vari{N}_c +\vari{N}_r (\vari{N}_c-D_c)\Big)\Big(P+\vari{W}_c\Big)} \; .
\]
Combining all the above results yields \eqref{eqtn:GaussianRCCR_Dr} and \eqref{eqtn:GaussianRCCR_Dcmax}.

%%%%%%%%%%%%%%%%%%%%%%%%%%%%%%%%%%%%%%%%%%%%%%%%%%%%%%%
\subsection{Proof of Lemma~\ref{lmma:GaussianSeparate}}
\label{subs:App_lmma:GaussianSeparate}
%%%%%%%%%%%%%

The Gaussian broadcast channel capacity is achieved by Gaussian $U_b$ and $U-U_b$ with $U_b\perp U-U_b$ (cf.~\cite{CoverThomas}). Let $0\leq\nu\leq 1$ and $\bar{\nu}=1-\nu$ control the power allocation between $U_b$ and $U-U_b$. The source rate-distortion function is similarly achieved by the test channel $X=Z_k+S_k$ with $Z_k\perp S_k$ for $k=b,g$. For these choices, \eqref{eqtn:SeparateCapacity1}, \eqref{eqtn:SeparateCapacity2}, \eqref{eqtn:SeparateRate1} and \eqref{eqtn:SeparateRate2} can be combined to give the following characterization of achievable distortions for general $\kappa$:

\begin{align}
\frac{\vari{N}_b}{D_b} &\leq \left(1+\frac{\nu P}{\bar{\nu}P+\vari{W}_b}\right)^{\kappa} \; , \label{eqtn:GaussianSeparateWorseChannel} \\
\frac{\vari{N}_b^2\vari{N}_g}{D_g\left[\vari{N}_g\vari{N}_b+D_b\left(\vari{N}_b-\vari{N}_g\right)\right]} &\leq \left(1+\frac{\nu P}{\bar{\nu}P+\vari{W}_b}\right)^{\kappa}\left(1+\frac{\bar{\nu}P}{\vari{W}_g}\right) ^{\kappa} &{\rm if\ } X-Y_g-Y_b \label{eqtn:GaussianSeparateBetterSideBetterChannel} \\
\frac{\vari{N}_g}{\min\left\{D_g,D_b+\frac{D_bD_g}{\vari{N}_b\vari{N}_g}\left(\vari{N}_g-\vari{N}_b\right)\right\}}
&\leq \left(1+\frac{\nu P}{\bar{\nu}P+\vari{W}_b}\right)^{\kappa}\left(1+\frac{\bar{\nu}P}{\vari{W}_g}\right) ^{\kappa} , &{\rm if\ }X-Y_b-Y_g\;. \label{eqtn:GaussianSeparateBetterSideWorseChannel}
\end{align}

%%%%%%%%%%%%%
The key to the proof is the observation that for optimal performance, \eqref{eqtn:GaussianSeparateWorseChannel} needs to be satisfied with equality for any $\kappa$.
To see this, assume that $(D_b,D_g)$ with $D_b<\mathbf{N}_b$ satisfies \eqref{eqtn:GaussianSeparateWorseChannel} with strict inequality for some $0<\nu\leq1$.
Then one can decrease $\nu$ until equality is obtained in \eqref{eqtn:GaussianSeparateWorseChannel}, and still satisfy \eqref{eqtn:GaussianSeparateBetterSideBetterChannel} or \eqref{eqtn:GaussianSeparateBetterSideWorseChannel}, depending on whether $X-Y_g-Y_b$ or $X-Y_b-Y_g$, respectively.
That, in turn, follows because the right-hand side of either of \eqref{eqtn:GaussianSeparateBetterSideBetterChannel} or \eqref{eqtn:GaussianSeparateBetterSideWorseChannel} are decreasing in $\nu$. Thus, if \eqref{eqtn:GaussianSeparateWorseChannel} is not tight, one can keep $D_b$ the same while decreasing $D_g$.

When $\kappa=1$, equality in \eqref{eqtn:GaussianSeparateWorseChannel} translates to 
\[
\bar{\nu} P = \frac{D_b(P+\mathbf{W}_b)}{\mathbf{N}_b} - \mathbf{W}_b \; .
\]
For the case $X-Y_g-Y_b$, \eqref{eqtn:GaussianSeparateBetterSideBetterChannel} then becomes
\[
D_g\geq \frac{\mathbf{N}_g\mathbf{N}_b^2 \mathbf{W}_g D_b}{\Big(D_b\mathbf{N}_b 
+\mathbf{N}_g (\mathbf{N}_b-D_b)\Big)\Big((\mathbf{W}_g-\mathbf{W}_b)\mathbf{N}_b+(P + \mathbf{W}_b)D_b \Big)} \; .
\]
If $X-Y_b-Y_g$, on the other hand, \eqref{eqtn:GaussianSeparateBetterSideWorseChannel} implies
\[
D_g \geq \frac{\mathbf{N}_g\mathbf{W}_g D_b}{ \Big((\mathbf{W}_g- \mathbf{W}_b)\mathbf{N}_b+(P+\mathbf{W}_b)D_b \Big)} 
\]
and
\[
D_g \geq \frac{\mathbf{N}_b\mathbf{N}_g \Big(\mathbf{N}_g\mathbf{W}_g  - (\mathbf{W}_g- \mathbf{W}_b)\mathbf{N}_b-(P+\mathbf{W}_b)D_b \Big)}{ \Big((\mathbf{W}_g- \mathbf{W}_b)\mathbf{N}_b+(P+\mathbf{W}_b)D_b \Big)\left(\mathbf{N}_g-\mathbf{N}_b\right)} 
\]
simultaneously, which is the desired result.

%%%%%%%%%%%%%%%%%%%%%%%%%%%%%%%%%%%%%%%%%%%%%%%%%%%%%
\subsection{Proof of Lemma~\ref{lmma:BinarySeparate}}
\label{subs:App_lmma:BinarySeparate}

For the binary symmetric channel, ${\cal C}(\kappa)$ is achieved by $U_b\sim \Ber(\tfrac{1}{2})$ and $U=U_b\oplus U_g$ with $U_g\sim\Ber(\theta)$ and $U_g$ independent of $U_b$.
The parameter $\theta$ serves as a tradeoff between $R_b$ and $R_g$.
The conditions \eqref{eqtn:SeparateCapacity1} and \eqref{eqtn:SeparateCapacity2} then become (cf.~\cite{CoverThomas})
\begin{eqnarray}
\label{eqtn:SeparateCapacityBinary1}
R_b & \leq & \kappa [1-H_2(\theta\star p_b)] \\
\label{eqtn:SeparateCapacityBinary2}
R_g & \leq & \kappa [H_2(\theta\star p_g) -H_2(p_g)] \; .
\end{eqnarray}

For the source coding part, we evaluate ${\cal R}^*(D_b,D_g)$ only with the auxiliary random variables chosen as in Section~\ref{subs:BinaryRCCR} where subscripts $c$ and $r$ are to be replaced by $r$ and $c$ or by $c$ and $r$.

These simple choices may {\em potentially} result in degradation of the separate coding performance, as the bounds on the alphabet sizes for ${\cal Z}_b$ and ${\cal Z}_g$ in~\cite{SteinbergMerhav,TianDiggavi,TianDiggavi2} are much larger.
However, our limited choice of $(Z_b,Z_g)$ can be justified in two ways: (i) to the best of our knowledge, there is no other choice known to achieve better rates, and (ii) to be fair, we use the same choice in our joint source-channel coding schemes.

As in the quadratic Gaussian case, we can write
\begin{equation}
\label{eqtn:MutualInformationBinary}
I(X;Z_k|Y_{k'}) = q_k r(\alpha_k,\beta_{k'})
\end{equation}
for $k,k'\in\{b,g\}$.
Combining \eqref{eqtn:SeparateRate1}, \eqref{eqtn:SeparateCapacityBinary1}, and \eqref{eqtn:MutualInformationBinary} yields \eqref{eqtn:BinarySeparateWorseChannel}.
Similarly, combining \eqref{eqtn:SeparateRate2}, \eqref{eqtn:SeparateCapacityBinary2}, and \eqref{eqtn:MutualInformationBinary}, we obtain
\eqref{eqtn:BinarySeparateBetterSideBetterChannel} when $X-Y_g-Y_b$, and \eqref{eqtn:BinarySeparateBetterSideWorseChannel} when $X-Y_b-Y_g$.

%%%%%%%%%%%%%%%%%%%%%%%%%%

\end{document}